\documentclass[11pt]{article}
\usepackage{amsthm,amsmath,amssymb,bbm}
\usepackage{natbib}
\usepackage{multirow}
\usepackage[pdftex]{graphicx}
\usepackage{subfigure}
\usepackage{makecell}
\usepackage{booktabs}
\usepackage{array}
\usepackage{url}
\usepackage{algorithm}
\usepackage{algorithmic}
\usepackage{bm}
\usepackage{smile}
\usepackage{wrapfig}
\usepackage{lipsum}
\usepackage{mathrsfs}
\usepackage{dsfont}
\usepackage{titling}
\usepackage{caption}
\usepackage[flushleft]{threeparttable}
\usepackage{tikz}

\usepackage{relsize}
\usepackage{rotating}
\usepackage[colorlinks,
linkcolor=red,
anchorcolor=blue,
citecolor=blue
]{hyperref}

\def\T{{ \mathrm{\scriptscriptstyle T} }}
\newcommand{\vecc}{{\textnormal{vec}}}
\newcommand{\cov}{\textnormal{cov}}
\newcommand{\vecl}{{\textnormal{vecl}}}

\newcommand{\dif}{{\textnormal{d}}}

\def\##1\#{\begin{align}#1\end{align}}
\def\$#1\${\begin{align*}#1\end{align*}}

\newcommand{\Rom}[1]{\text{\uppercase\expandafter{\romannumeral #1\relax}}}
\newcommand{\pr}{\textnormal{pr}}
\newcommand{\var}{\textnormal{var}}

\usepackage{geometry}
\usepackage{enumitem}
\begin{document}
\title{  {\LARGE Counting Process Based Dimension Reduction Methods for Censored Outcomes}\thanks{The first two authors contribute equally.}}
%\author{Qiang Sun$^*$, ~~Ruoqing Zhu$^*$,~~ Tao Wang,~~ Yuanjia Wang,~~ Donglin Zeng}

\author{
Qiang Sun\thanks{Department of Statistical Sciences, University of Toronto, Toronto,  ON M5S 3G3; e-mail: \href{mailto:qsun@utstat.toronto.edu}{qsun@utstat.toronto.edu}. }
~~Ruoqing Zhu\thanks{The corresponding author. Department of Statistics, University of Illinois at Urbana-Champaign, Champaign, IL 61820; email: \href{mailto:rqzhu@illinois.edu}{rqzhu@illinois.edu}.}
~~Tao Wang \thanks{Department of Bioinformatics and Biostatistics, Shanghai Jiao Tong University, Shanghai, China; email: \href{mailto: neowangtao@sjtu.edu.cn}{neowangtao@sjtu.edu.cn}.}
~~Donglin Zeng\thanks{Department of Biostatistics,  University of North Carolina at Chapel Hill, Chapel Hill, NC 27599; email:  \href{mailto:dzeng@email.unc.edu}{dzeng@email.unc.edu}.}
}

\date{\today}

\maketitle

%\vspace{-0.5in}

%\newpage
% typeset the title of the contribution

%\setcounter{page}{1}
%\centerline{\LARGE Counting Process Based Dimension Reduction for Survival Data}\vspace{1cm}

\begin{abstract}
We propose a class of dimension reduction methods for right censored survival data using a counting process representation of the failure process. Semiparametric estimating equations are constructed to estimate the dimension reduction subspace for the failure time model. The proposed method addresses two fundamental limitations of existing approaches. First, using the counting process formulation, it does not require any estimation of the censoring distribution to compensate the bias in estimating the dimension reduction subspace. Second, the nonparametric part in the estimating equations is adaptive to the structural dimension, hence the approach circumvents the curse of dimensionality. Asymptotic normality is established for the obtained estimators. We further propose a computationally efficient approach that simplifies the estimation equation formulations and requires only a singular value decomposition to estimate the dimension reduction subspace.  Numerical studies suggest that our new approaches exhibit significantly improved performance for estimating the true dimension reduction subspace. We further conduct a real data analysis on a skin cutaneous melanoma dataset from The Cancer Genome Atlas. The proposed method is implemented in the R package ``orthoDr''.
\end{abstract}

\noindent {\bf keywords}
Sufficient Dimension Reduction; Survival Analysis; Estimating Equations; Semiparametric Inference; Sliced Inverse Regression.

\section{Introduction}\label{sec:1}

Dimension reduction is an important problem in regression analysis. It aims to extract a low-dimensional subspace from a $p$-dimensional covariates $X=(X_1,\ldots, X_p)^\T$, to predict an outcome of interest $T$. The dimension reduction literature often assumes the multiple-index model
\vspace{-0.15in}
\#\label{eq:1}
T=h\big(B^\T X, \epsilon\big),
\vspace{-0.15in}
\#
where $\epsilon$ is a random error independent of $X$, $B\in\RR^{p\times d}$ is a coefficient matrix with $d < p $, and $h(\cdot)$ is a completely unknown link function. This model is equivalent to assuming $T\perp X \mid B^\T X$ \citep{li1991sliced}. Since any $d$ linearly independent vectors in the linear space spanned by the columns of $B$ also satisfy model (\ref{eq:1}) for some  $h$, we define $\cS(B)$ to be this linear subspace. Furthermore, we call the intersection of all such subspaces satisfying $T\perp X \mid B^\T X$ the central subspace, denoted by $\cS_{T|X}$, whose dimension is called the structural dimension. According to \citet{cook2009regression}, $\cS_{T|X}$ is uniquely defined under mild conditions. The goal of sufficient dimension reduction in \eqref{eq:1} is to determine the structural dimension and the central subspace using empirical data.

There is an extensive literature on estimating the central subspace for completely observed data, including the seminal paper from \cite{li1991sliced} and subsequent works such as \cite{cook1991discussion}, \cite{zhu2006sliced}, \cite{li2007directional}, \cite{xia2007constructive}, and \cite{ma2012semiparametric}. When $T$ is subject to right censoring, which frequently occurs in survival analysis, model (\ref{eq:1}) includes many well-known survival models as special cases, for instance, the Cox proportional hazard model \citep{cox1972regression}, the accelerated failure-time model \citep{lin1998accelerated}, and the linear transformation models \citep{zeng2007maximum}.

There has been limited work on estimating the dimension reduction subspace using censored observations. \cite{li1999dimension} propose a modified sliced inverse regression method which uses the estimate of the conditional survival function to account for censored cases. \cite{xia2010dimension} propose to estimate the conditional hazard function nonparametrically and utilize its gradient and local linear regression to construct the dimension reduction directions. In \cite{li1999dimension}, a $p$ dimensional kernel estimation is used for compensating the bias caused by censoring, while in \cite{xia2010dimension}, the estimation procedure also requires a $p$ dimensional kernel hazard function to provide reliable initial value, and then gradually reduce the dimension to $d$. Hence these methods inevitably suffer from the curse of dimensionality. When $p$ is not small, some alternative approaches such as \cite{lu2011sufficient} adopt an inverse probability weighting scheme, which implicitly requires the correct specification of the censoring mechanism.

In this paper, we propose a counting process-based dimension-reduction framework that leads to four different approaches. The proposed methods enjoy unique advantages and address several limitations of the existing literature. First, our approach is built upon a counting process representation of the underlying survival model. This allows a construction of doubly robust estimating equations, and the resulting estimator is more stable than existing approaches such as \cite{xia2010dimension}. This formulation can avoid the linearity assumption \citep{li1991sliced} and the estimation of any censoring distribution, which are necessary components of \cite{li1999dimension} and \cite{lu2011sufficient}. Second, the proposed framework is adaptive to the structural dimension in the sense that the involved nonparametric estimations only depend on the dimension of $\cS(B)$, which is usually small, thus circumvents the curse of dimensionality. To this end, the proposed method shares similar advantage as \cite{xia2010dimension}. However, computationally, we utilize an optimization approach on the Stiefel manifold \citep{wen2013feasible} to solve the estimating equations, which is numerically stable and fast. Last, under some restrictive assumptions, our method reduces to a computationally efficient approach that can directly estimate the dimension reduction subspace without nonparametric smoothing.

\section{Proposed methods} \label{sec:method}

\subsection{Semiparametric estimating equations for the central subspace}
Throughout the paper, we denote the failure time by $T$ and the censoring time by $C$. Let $Y \!=\! \min(T, C)$ and $\delta = I(T \leq C)$ be the observed event time and the censoring indicator, respectively. We assume that $C$ is independent of $T$ conditional on $X$.
Let
$N(u) \!=\!  I( Y\!\leq\! u ,\delta\!=\!1)$
and
$Y(u) \!=\! I(  Y \!\geq\! u)$
denote the observed counting process and the at-risk process, respectively. Let $\lambda(u|X)$ be the conditional hazard  for $T$ given $X$. Due to \citet{xia2010dimension}, model (\ref{eq:1}) is equivalent to $\lambda({u|X}) = \lambda({u|B^\T X})$. We further let $dM(u,X)\!=\!dM(u, B^\T X)\!=\!dN(u) -\lambda(u|B^\T X)Y(u)du$ be the martingale increment process indexed by $u$. This paper centers on constructing estimation equations that are built upon the counting process representation of the survival model. To derive the estimating equations, we follow \cite{bickel1993efficient} and \cite{tsiatis2007semiparametric} to obtain the ortho-complement of the nuisance tangent space at $B$ as
\#\label{eq:3.4}
\cE^\perp&=\bigg\{ \int \big\{\alpha(u, X)-\alpha^*(u, B^\T X)\big\} dM(u,X): \nonumber \\
&\qquad\quad\alpha(u,X) \,\,\text{is a measurable function of $X$ and $u$} \bigg\},
\#
where
\$
\alpha^*(u,B^\T X) = E\big\{\alpha(u,X)\big| \cF_{u} , B^\T X\big\},
\$
and $\cF_{u}$ the filtration. The derivation can be found in the Supplementary Material. To estimate $B$, we consider the unbiased estimating equations
\#\label{eq:eepop}
E\Big[ \int \big\{\alpha(u, X)-\alpha^*(u, B^\T X)\big\} \big\{dN(u)-\lambda(u|B^\T X)Y(u)du\big\} \Big] = 0.
\#
The sample versions based on $n$ independent and identical copies $\{Y_i, \delta_i, X_i\}_{i=1}^n$ are given by
\#\label{eq:ee}
n^{-1}\sum_{i=1}^n\Big[\int \big\{\alpha(u, X_i)-\alpha^*(u, B^\T X_i)\big\} \big\{dN_i(u)-\lambda(u|B^\T X_i)Y_i(u)du\big\}\Big]=0,
\#
where the conditional hazard function will be estimated using the data. For some particular choices of $\alpha(u, X)$, this can be implemented utilizing the generalized method of moments \citep{hansen1982large}:
\#\label{eq:GMM}
B=\underset{B\in \Theta}{\arg\min} \,\, \big\{ {\psi}_n(B)^\T {\psi}_n(B) \big\},
\#
where ${\psi}_n(B)$ is the left hand side of \eqref{eq:ee}. Several quantities in ${\psi}_n(B)$ need to be estimated nonparametrically. For example, the conditional hazard function $\lambda(u|B^\T X_i)$ at any time point $u$ can be estimated by
\#
\widehat \lambda(u | B^\T X=z) &= \frac{\sum_{i=1}^n K_b(Y_i-u)\delta_iK_h\big(B^\T X_i - z\big)}{\sum_{j=1}^n I\big(Y_j \geq u\big) K_h\big(B^\T X_j - z\big)}, \label{eq:chf}
\#
for some bandwidths $b$, $h$, and kernel function $K_h(\cdot) = h^{-1}K(\cdot/h)$. Details of these nonparametric components will be deferred to Section \ref{sec:alg}. It is worth noting that this nonparametric component requires is only a $d$ dimensional kernel, hence adaptive to the underlying structure.

It is then crucial to choose specific forms of $\alpha(u, X)$. Different choices may result in simplifications of the above formulation and/or gain additional theoretical and computational advantages. In the next two sections, we present four different choices, which fall into two categories: the forward and inverse regression schemes. The main differences between the two schemes are whether the counting process $N(u)$ is used in the definition of $\alpha(u, X)$. The forward regression scheme is essentially the estimating equations approach in the normal regression, while the inverse regression scheme utilizes $N(u)$ to mimic the sliced inverse regression \citep{li1991sliced} conceptually.

\subsection{Forward regression} In the forward regression scheme, we choose $\alpha(u, X)$ such that it does not depend on the observed failure process $N(u)$. We first notice that, as long as $\alpha(u, X)$ depends at most on the at-risk process $Y(u)$, we can simplify the estimating equations in \eqref{eq:eepop} into
\# \label{eq:ee-FR}
E\bigg({{{\int}}} \Big[\alpha(u, X) -E\big\{\alpha(u, X)|Y(u)=1, B^\T X\big\} \Big] dN(u)\bigg) = 0.
\#

We now give one example of $\alpha(u, X)$ in the following when the structural dimension $d=1$. This requires only a $1$-dimensional nonparametric estimation.
\begin{example}\label{ex:FR}
With $\alpha(u, X)=X$, the population version of the $p$-dimensional estimating equations is given by:
\#\label{eq:FR1}
E\bigg({{{\int}}} \Big[X-E\big\{X|Y(u)=1, B^\T X\big\} \Big] dN(u)\bigg)=0.
\#
This formulation reduces to the set of efficient estimating equations for the Cox proportional hazard model when the exponential link is known to be the underlying truth. It can also be used for the transformation models proposed by \cite{zeng2007maximum}. For some simple extensions, we could let $\alpha(u, X) = E\{X Y(u)\} X^\T$ to obtain $p$-by-$p$ dimensional estimating equations, which is suitable for the case of $d>1$. To implement the forward regression method given by \eqref{eq:FR1}, noticing that $dN_i(u)$ takes a jump at $Y_i$ only if $\delta_i = 1$, we can estimate $\psi_n(B)$ in \eqref{eq:GMM} using
\#
\widehat \psi_n\big(B\big) &= \frac{1}{n}\sum_{i=1}^n \big\{ X_i - \widehat E(X | Y \geq Y_i, B^\T X_i ) \big\} \delta_i, \label{eeq:FR1}
\#
 where $\widehat E\{X | Y \geq u, B^\T X = z\}$ takes the following form for any given $u$ and $z$,
\#
\frac{\sum_{i=1}^n X_i I\big( Y_i \geq u \big) K_h \big(B^\T X_i - z \big) }{\sum_{i=1}^n I\big( Y_i \geq u \big) K_h \big(B^\T X_i - z \big)}, \label{eq:condx}
\#
for some choices of bandwidth $h$ and kernel function $K_h(\cdot)$. Again, these details are deferred to Section \ref{sec:alg}.
\end{example}

\subsection{Inverse regression}\label{sec:IRS}

In this section, we focus on the inverse regression scheme. An important property that motivates the development is
\#\label{eq:local}
\big\{ d N(u) \,|\, Y(u)=1, B^\T X\big\}
&\sim \text{Bernoulli} \big\{ \lambda(u|B^\T X)du \big\},
\#
where $dN(t)=N(t+dt)-N(t)$. Hence, we can consider the sliced conditional mean of $X$ given the outcome of $dN(t)$ among the risk set, i.e., $Y(t) = 1$. This leads to the construction of a local mean difference that is essentially the sliced mean difference for the binary outcome $dN(u)$ \citep{cook1999dimension}:
\#\label{eq:phiu}
\varphi(u) = E\big\{X\big|dN(u)\!=\!1, Y(u)\!=\!1\big\}\!-\!E\big\{X\big|dN(u)\!=\!0, Y(u)\!=\!1\big\}.
\#
It should be noted that the outcome $dN(u)$ conditioning on the event $Y(u) = 1$ depends only on the failure model $\lambda(u|B^\T X)$ \citep{xia2010dimension}. Hence, by varying the argument $u$, the inverse regression curve $\varphi(u)$ is contained within the central subspace $\cS_{T|X}$. With this $\varphi(u)$ established, we consider the function
\#\label{eq:alpha}
\alpha(u,X) = X \varphi^\T(u). %= X \varphi^\T (u).
\#
Then,
\#\label{eq:alpha2}
\alpha(u,X) - \alpha^\ast(u, B^\T X)  = \big[ X - E\{X | Y(u) = 1, B^\T X\} \big] \varphi^\T(u).
\#
This particular choice can be implemented by estimating $E\{X | Y(u) = 1, B^\T X\}$ using equation \eqref{eq:condx} and estimating $\varphi(u)$ using the sliced average based on \eqref{eq:phiu}:
\# \label{eq:phi}
\widehat \varphi(u)=\frac{\sum_{i=1}^n X_i I\big( u \leq Y_i < u + h, \delta_i = 1\big) }{\sum_{i=1}^n I\big( u \leq Y_i < u+h, \delta_i = 1\big)} - \frac{\sum_{i=1}^n X_i I\big(Y_i \geq u\big)}{\sum_{i=1}^n I\big(Y_i \geq u\big)}.
\#
Based on this choice of $\alpha$, we propose two methods that utilize the estimating equations \eqref{eq:eepop}, and a computationally efficient method that further simplifies the formula to a singular value decomposition problem.

\begin{example}\label{ex:Semi}
Replacing $\alpha(u,X) - \alpha^\ast(u, B^\T X)$ in \eqref{eq:eepop} by \eqref{eq:alpha2} leads to the semiparametric inverse regression approach in its population version:
\#\label{eq:semi}
E \bigg( \int \big[ X - E\{X | Y(u) = 1, B^\T X\} \big] \varphi^\T(u) dM(u) \! \bigg).
\#
This consists of $p \times p$ estimating functions, and is able to handle $d > 1$. However, the nonparametric estimation part is only $d$ dimensional as reflected by $B^\T X$. Furthermore, this formulation enjoys the double robustness property which is illustrated in the Supplementary Material. Similar phenomenon has been observed by \cite{ma2012semiparametric} in the regression setting but without censoring. This suggests that if one of the terms $E\{X | Y(u) = 1, B^\T X\}$ and $M(u)$ is estimated incorrectly, we can still obtain consistent estimations of the dimension reduction subspace. In our numerical experiment, we do observe numerical advantage of this approach over its simplified version, which is given in Example \ref{ex:CP}.
\end{example}

To implement this method, a vectorized $\psi_n(B)$ is given by
\#\label{eeq:IR-Semi}
\!\!\widehat\psi_n\big(B\big)\!=\!\vecc \Bigg[
\!\frac{1}{n}\sum_{i=1}^n \sum_{\substack{j=1 \\ \delta_j = 1}}^n \!\left\{ X_i \!-\! \widehat E \big(X \big| Y \geq Y_j, B^\T X_i\big) \right\} \widehat \varphi^\T(Y_j) \left\{\delta_i I(j \!=\! i) \!-\! \widehat \lambda\big(Y_j | B^\T X_i\big) \right\} \Bigg],
\#
where $\widehat E\{X | Y \geq u, B^\T X = z\}$ and $\widehat \varphi^\T(u)$ are given in \eqref{eq:condx} and \eqref{eq:phi}, respectively, and the conditional hazard function can be estimated by \eqref{eq:chf}. We finally apply the generalized method of moments \eqref{eq:GMM} to estimate $B$.

\begin{example}\label{ex:CP}
Similar to the forward regression example, our choice of $\alpha(u,X)$ in \eqref{eq:alpha} depends on at most the at-risk process $Y(u)$. Hence, the estimating functions in \eqref{eq:semi} can be simplified to the following counting process inverse regression approach:
\#\label{eq:dN}
E \bigg( \int \big[ X - E\{X | Y(u) = 1, B^\T X\} \big] \varphi^\T(u) dN(u) \! \bigg).
\#

Replacing $dM(u)$ with $dN(u)$ greatly reduces the computational burden. This can be seen from \eqref{eeq:IR-Semi}, where a conditional hazard function $\widehat \lambda\big(Y_j | B^\T X_i\big)$ needs to be evaluated at each observed failure time point $j$ for all observations $i$. Of course, by doing this simplification, we lose the double robustness property. The implementation of this approach is a simplified version of \eqref{eeq:IR-Semi} with:
\#\label{eeq:IR-CP}
\!\!\widehat \psi_n\big(B\big)\!=\!\vecc \Bigg[
\!\frac{1}{n}\sum_{i=1}^n \!\left\{ X_i \!-\! \widehat E \big(X \big| Y \geq Y_i, B^\T X_i\big) \right\}\delta_i \widehat \varphi^\T(Y_i) \Bigg],
\#
where the estimations of nonparametric components are provided previously.
\end{example}

\begin{example}\label{ex:CPSIR} With some additional assumptions, $B$ can be estimated without any nonparametric smoothing. We make the following definitions:

\begin{definition}%[Local Linearity Condition and Local Proportional Constant Variance Condition]
For any $\alpha \in \mathbb{R}^p$ and any $u >0$, the linearity condition \citep{li1991sliced} is satisfied further conditioning on the event $\{Y(u) = 1\}$, i.e.,
\#
E\{\alpha^\T X | Y(u) = 1, B^\T X = z\} = c_0(u) + c^\T(u) z,
\#
where $c_0(u)$ and $c(u)$ are constants that possibly depend on $u$. Furthermore, the time-invariant covariance condition requires
\#
\text{Cov}\{X | Y(u) = 1\} = c_1(u) \Sigma,
\#
where $c_1(u)$ is some constant that depends on $u$.
\end{definition}
Noticing that after centering $X$ at time point $u$, if the above two conditions are satisfied, we have
\$
&E(X| Y(u) = 1, B^\T X) - E(X| Y(u) = 1) \nonumber \\
&~= P \big\{ X - E(X| Y(u) = 1) \big\},
\$
where $P=\Sigma B(B^\T\Sigma B)^{-1}B^\T$ and the constant term $c_1(u)$ vanishes. Realizing that by the time-invariant covariance condition, $P$ remains the same across all time points, plugging in the above equation into \eqref{eq:dN} leads to
\$
&Q~E\bigg(\! \int \big[ X - E\{X| Y(u) = 1\} \big] \varphi^\T(u) dN(u)\! \bigg) \!=\!0,
\$
where $Q=I-P$. This is equivalent to deriving the left-singular space of the covariance matrix
\#\label{eq:CPSIR}
E\bigg(\! \int \big[ X - E\{X| Y(u) = 1\} \big] \varphi^\T(u) dN(u)\! \bigg).
\#
\end{example}
The computation of this approach is extremely simple. Realizing that $dN(u)$ takes value 1 at at most one time point on the entire time domain, which corresponds to the failure subjects, the covariance form can be estimated by a sum of $n$ terms. Then we perform singular value decomposition on this sample covariance matrix and obtain its leading left singular vectors, hence no optimization is required. Details are provided in Algorithm \eqref{alg:cpsir}.

\begin{remark}
The two conditions imposed in this example are somehow restrictive and do not always hold. For example, since $Y(u)$ is a process that depends on both the failure and censoring distribution, as long as the censoring distribution depends on structures beyond $B^\T X$, the conditions could be violated. However, many recent works of literature argue that the sliced inverse regression seems to still have satisfactory performances even when the linearity condition do not hold \citep{li2009dimension, dong2010dimension}. Hence, this does not prevent the method from serving as a good explorative tool. The method is also practically very useful since it is served as the initial value when solving our other optimization approaches to speed up the computation.
\end{remark}

\section{Implementation and Algorithm}\label{sec:alg}

The implementation of the computationally efficient method given in \eqref{eq:CPSIR} is straightforward since only sliced averaging and eigen-decomposition are required. Algorithm \eqref{alg:cpsir} summarizes the estimation procedure.

\begin{algorithm}[!h]
\caption{Algorithm for the computationally efficient approach.}\label{alg:cpsir}
\vspace*{-12pt}
\begin{tabbing}
   \enspace Input: $\{(X_i, \delta_i, Y_i), 1\leq i\leq n\}$, $h > 0$, $k>0$.\\
   \enspace Step 1: For each $Y_i$ such that $\delta_i = 1$, calculate $\widehat \varphi(Y_i)$ using Equation \eqref{eq:phi} and calculate\\
   \qquad $\widehat E(X | Y > Y_i)$ using $\widehat E(X | Y > u) = \textstyle \{ \sum_{i=1}^n I( Y_i > u ) \}^{-1} \{\sum_{i=1}^n X_i I( Y_i > u )\}.$\\
   \enspace Step 2: Calculate $\widehat{M} = n^{-1} \sum_{\delta_i = 1} \{ X_i - \widehat E(X | Y_i) \} \widehat \varphi^\T(Y_i)$.\\
   \enspace Step 3: Perform the singular value decomposition: $\widehat{M} = \widehat U\widehat D\widehat V^\T$.\\
\enspace Output: $\widehat B$ as the first $k$ columns of $\widehat U$.
\end{tabbing}
\end{algorithm}

It requires numerical optimization to solve the estimating equations of the forward regression approach given in \eqref{eq:FR1} and the two inverse regression approaches, given in \eqref{eq:semi} and \eqref{eq:dN}, respectively. For all three approaches, we use the corresponding choice of the moment conditions and solve for the minimizer of ${\widehat \psi}_n(B)^\T {\widehat \psi}_n(B)$, where ${\widehat \psi}_n(B)$ is specified in \eqref{eeq:FR1}, \eqref{eeq:IR-Semi} and \eqref{eeq:IR-CP} respectively. Existing methods use general-purpose optimization tools such as the Newton--Raphson to solve for the minimizer, however, dimension reduction methods create an additional difficulty due to the identifiability issue, i.e., $B$ is not uniquely defined and the rank may not be preserved if we solve it freely within the space of $\mathbb{R}^{p\times d}$. To tackle this, \cite{ma2012semiparametric} propose to set the upper block (or a selected set of $d$ rows) of $B$ as a diagonal matrix and solve for the rest of parameters. However, this requires the pre-knowledge of the location of the important variables. Instead, we propose an orthogonality constrained optimization approach to solve our semiparametric dimension reduction model within the Stiefel manifold \citep{edelman1998geometry}:
\#\label{optm:B}
\text{minimize} \quad & {\widehat \psi}_n(B)^\T {\widehat \psi}_n(B), \nonumber \\
\text{subject to} \quad & B^\T B = I_{d \times d}.
\#
The advantage of this optimization approach is that we exactly preserves the rank $d$ of the column space defined $B$ while not pre-specify the restrictions on any of it entries. The main machinery of this algorithm is the optimization approached proposed by \cite{wen2013feasible}. The method is a first-order descent algorithm that preserves the update of the parameters within the manifold. In particular, let the gradient matrix be defined as
\#\label{eq:gradient}
G = \frac{\partial \, {\widehat \psi}_n(B)^\T {\widehat \psi}_n(B) }{\partial B }.
\#
Then, utilizing the Cayley transformation, we can update $B$ to
\#\label{eq:update}
B(\tau_0) &= \Big(I + \frac{\tau_0}{2} A \Big)^{-1} \Big(I - \frac{\tau_0}{2} A \Big) B,
\#
where $A = G B^\T - B G^\T$ is a skew-symmetric matrix, and $\tau_0$ is a step size. In practice, $\tau_0$ can be chosen using inexact line search by incorporating the Wolfe conditions \citep{nocedal2006sequential}. It can be easily verified that if $B^\T B = I$, then $B(\tau_0)^\T B(\tau_0) = I$ for any $\tau_0> 0$, which preserves the constraint exactly. This approach is in-line with traditional dimensional reduction methods which recover the column space of $B$ rather than treating each entry as a fixed parameter. Moreover, if an upper block diagonal version is desired, we can easily convert the obtained solution through linear transformations. However, in this case, we can select the largest entries in the estimated $\widehat B$ as the location of the diagonal matrix, instead of pre-specifying the locations. The full algorithm is presented in Algorithm \ref{alg:2}. The iteration is stopped when a pre-specified optimization precision $\varepsilon_0$ is reached. For estimating the nonparametric components \eqref{eq:chf} and \eqref{eq:condx}, we exploit the Gaussian kernel and choose the optimal bandwidth $h = \big(4/(d+2)\big)^{1/(d+4)}n^{-1/(d+4)}\widehat\sigma$ \citep{silverman1986density}, where $\widehat\sigma$ is the estimated standard deviation. For estimating the conditional hazard function, equation \ref{eq:chf} bears much computational burden because it requires ${\cal O}(n^2)$ flops to calculate the hazard at any given $u$ and $z$. An alternative approach that greatly reduces the computational cost can be considered using the definition in \cite{dabrowska1989uniform}, given by the following:
\#
\widehat \lambda(u | B^\T X=z) &= \frac{\sum_{i=1}^n I\big(Y_i = u\big)I\big(\delta_i = 1\big) K_h\big(B^\T X_i - z\big)}{\sum_{j=1}^n I\big(Y_j \geq u\big) K_h\big(B^\T X_j - z\big)} \label{eq:chf2}
\#
Since the indicator $I\big(Y_i = u\big)$ can only take 1 if $u$ is among the observed survival times. Hence, the numerator essentially requires only a single flop. Based on our experience, the numerical performance of the two versions are very similar. Hence, the above definition is implemented and used in the simulation study. Lastly, the implementation is available through the R package ``orthoDr'' \citep{zhao2017orthoDr} through the Rcpp \citep{eddelbuettel2011Rcpp} interface.

\begin{algorithm}[!h]
\caption{The orthogonality constrained optimization algorithm.}\label{alg:2}
\vspace*{-12pt}
\begin{tabbing}
   \enspace Input: $\varepsilon_0, ~\{(X_i, \delta_i, Y_i), 1\leq i\leq n\}$.\\
   \enspace Initialize: Obtain $B^{(0)}$ from the computationally efficient approach in Algorithm \ref{alg:cpsir}.\\
   \enspace For $k = 1$ to $k =$ max.iter: \\
   \qquad Numerically approximate the gradient matrix G at $B^{(k)}$. \\
   \qquad Compute the skew-symmetric matrix matrix $A = G B^\T - B G^\T$. \\
   \qquad Perform line search for $\tau_0$ on the path $B(\tau_0) = \big(I + \frac{\tau_0}{2} A \big)^{-1} \big(I - \frac{\tau_0}{2} A \big) B$. \\
   \qquad Update $B^{(k+1)} = B(\tau_0)$.\\
   \qquad Stop if $\big\|B^{(k+1)}-B^{(k)}\big\|_2\leq \varepsilon_0$.\\
\enspace Output: $\widehat B=B^{(k+1)}$.
\end{tabbing}
\end{algorithm}

\section{Asymptotic Normality}\label{sec:theory}

We focus on the semiparametric inverse regression approach, in which $\widehat B$ obtained by solving
\$
\frac{1}{n}\vecc{\bigg[\sum_{i=1}^n \int_0^\tau \Big\{ \alpha(u,X_i) - \widehat\alpha^*(u, \widehat B^\T X_i)\Big\}d\widehat M(u, \widehat B^\T X_i)\bigg]}=0,
\$
To address the identifiability issue of $B$, we restrict our attention to the matrices in the form of  $B=(B_u^\T, B_\ell^\T)^\T$, where the upper sub-matrix $B_u=I_d\in \RR^{d\times d}$ is the $d\times d$ identity matrix. In this manner, we can view $B_\ell$ as the unique parameterizations of the subspace $\cS(B)$. We then write $\beta_\ell=\vecl(B)=\vecc(B_\ell)$, the vector concatenating all free parameters in $B$.     We need the  the following regularity assumptions.

\begin{assumption}\label{ass:1}
There exists a $\tau$, such that $0<\tau<\infty$ and $\textnormal{pr}(Y>\tau|X)>0$.
\end{assumption}

%\begin{assumption}\label{ass:2}
%The conditional  hazard function, $\lambda(u, z)$, is bounded away from $0$ and $\infty$ uniformly for all $u$ and $z$.
%\end{assumption}
%\begin{assumption}\label{ass:3}
%\comment{Q: revise this assumption.}
%There exists a predictable and square-integrable function $\Gamma(u)$ with respect to $\Lambda_i(u),~i=1,\ldots, n$ , such that, for any $X$ and $\widehat B\in \Omega(B)$ and $0\leq u\leq \tau$, we have
%\begin{gather*}
%\Big|\vecc\big\{\alpha^*(u,\widehat B^\T X)\big\}\!-\!\vecc\big\{\alpha^*(u, B^\T X)\big\}\!-\!\big[{\nabla_{\beta_\ell}}\vecc \big\{\alpha^*(u, B^\T X)\big\}\big] (\widehat\beta_\ell\!-\!\beta_\ell)\Big|
%\!\leq\! \Gamma(u)\big\|\widehat\beta_\ell\!-\!\beta_\ell\big\|_2^2,\\
%~\textnormal{and}~ \Big|\lambda(u,\widehat B^\T X)-\lambda(u, B^\T X)-\big\{{\nabla_{\beta_\ell}}\lambda(u, B^\T X)\big\}^\T \big(\widehat\beta_\ell-\beta_\ell\big)\Big|
%\leq \Gamma(u)\big\|\widehat\beta_\ell-\beta_\ell\big\|_2^2.
%\end{gather*}
%\end{assumption}

\begin{assumption}\label{ass:pdf}
Let $f_{B^\T X}(z)$ be the density  function of $B^\T X$ evaluated at $z=B^\T x$, $f(t, z)$ be the density of $T$ given $B^\T X=z$, $S(t, x)=\pr(T\geq t| X=x)$ and $S_c(t, x)=\pr(C\geq t | X=x)$.  Assume that $f(t,z)$, $f_{B^\T X}(z)$, $S(t, z)$ and $E\big(S_c(t, X)|z\big)$ are bounded, and have bounded  first and second derivatives with respect to $t$ and  $z$, and  $S(t, z)$ is bounded away from zero.
\end{assumption}

\begin{assumption}\label{ass:kernel}
The univariate kernel function $K(x)$ is symmetric with $\int x^2K(x)\dif x<\infty$. The $d$-dimensional kernel function is a product of $d$ univariate kernel functions, that is $K(u)=\prod K(u_j)$ for $u=(u_1, \ldots, u_d)^\T.$
\end{assumption}

Assumptions \ref{ass:1} and \ref{ass:pdf} are standard in survival analysis. Assumptions \ref{ass:kernel} is commonly used in kernel estimations. Base on these two assumptions, we can provide the rate for our conditional hazard function estimation and its derivatives. Its easy to see that the Silverman formula implemented in our numerical approach automatically leads to consistent estimations.

\begin{lemma}\label{lem:thm.1}
Under Assumption \ref{ass:kernel} and \ref{ass:pdf}, and assume that the bandwidths satisfy $h, b \rightarrow 0$, $nbh^{d+2} \rightarrow \infty$, we have, uniformly for all $t$ and $z$,
\$
%\widehat\Lambda(t| z)&=\Lambda(t,z)+O_p\left(\big(nh^d\big)^{-1/2}+h^2\right), ~~\textnormal{and}\\
\widehat\lambda(t| z)&=\lambda(t, z)+O_p\left(\big(nbh^d\big)^{-1/2}+h^2+b^2\right),~\textnormal{and}~\\
\frac{\partial}{\partial z}\widehat\lambda(t| z)&=\frac{\partial}{\partial z} \lambda(t, z)+O_p\left(\big(nbh^{d+2}\big)^{-1/2}+h^2+b^2\right).
\$
\end{lemma}

Before presenting our main theorem, we also need the convergence result of the $\alpha^*$ functions. However, we do not want the theoretical result being limited to the choice given in equation \eqref{ex:Semi}. Instead, we provide general results for any valid choice of the $\alpha^*$ function, as long as the following condition is satisfied.

\begin{assumption}\label{ass:4}
We assume that for some $\kappa<1/2$, the convergence rate for the following conditional nonparametric estimation  holds uniformly over all $u$ and $z$,
\begin{gather*}
\vecc \Big\{\widehat\alpha^*\big(u, z\big)-\alpha^*\big(u,z\big)\Big\}  ={\rm O}_p\left(n^{-{1}/{2}+\kappa}\right),\\
\frac{\partial }{\partial z}\vecc \Big\{\widehat\alpha^*\big(u, z\big)-\alpha^*\big(u,z\big)\Big\} ={\rm O}_p\left(n^{-{1}/{2}+\kappa}\right),
%\widehat\Lambda\big(u, z\big)-\Lambda\big(u, z\big)  ={\rm O}_p\big(n^{-1/2+\kappa}\big),~~~\widehat\lambda\big(u, z\big)-\lambda\big(u, z\big)  ={\rm O}_p\big(n^{-1/2+\kappa}\big),~\textnormal{and}\\
%\frac{\partial}{\partial z}\left\{\widehat\lambda\big(u, z\big)-\lambda\big(u, z\big)\right\}  ={\rm O}_p\big(n^{-1/2+\kappa}\big).
\end{gather*}
\end{assumption}

Note that, for most valid choices such as a kernel estimation of the conditional density, when the number of dimension $d$ is fixed, the rate provide in Lemma \ref{lem:thm.1} is essentially valid for $\widehat\alpha^*\big(u, z\big)$, while for conditional expectation estimations, the classical rate of $O_p\big((nh^d)^{-1/2}+h^2\big)$ can obtained. Hence, with proper choice of the bandwidth, the rate of Assumption \ref{ass:4} can always be achieved. We now present the main theorem.

\begin{theorem}[Asymptotic Normality]\label{thm1}
Under Assumptions \ref{ass:1}-\ref{ass:4}, and the choice of bandwidths specified in Lemma 1, the estimator $\vecl (\widehat B )$ %obtained from the estimating equation \eqref{eq:7}
is asymptotically normal, that is
\$\sqrt{n} \, \left(\widehat\beta_\ell-\beta_\ell\right)  \xrightarrow{d} \cN (0,\Sigma),
\$ where $\Sigma=(G^\T G)^{-1} G\Sigma_AG^\T (G^\T G)^{-1}$,
\begin{gather*}
\Sigma_A=\cov \big(A(\tau)\big)=\cov\left(\int_0^\tau\vecc\left\{\alpha(u, X)-\alpha^*(u,  B^\T X)\right\}dM(u, B^\T X)\right),\\
\textnormal{and}\,\, G=E\left(\frac{\partial }{\partial \beta_\ell}{\int_0^\tau \vecc{\left[ \big\{\alpha(u, X_i)- \alpha^*(u,  B^\T X) \big\}\right]} d  M(u,  B^\T X)   }\right).
\end{gather*}

%\comment{Q: define $\Sigma$ here.}
%\begin{gather*}
%\sqrt{n} \textnormal{vec} \big(\widehat B -B\big)  \xrightarrow{d} \cN (0,\Sigma), ~\textnormal{for some}~\Sigma.%\\
%\Sigma=\frac{1}{4}\bigg\{E \int_0^\tau G(u, B^\T X)du\bigg\}^{-1}\bigg\{ E\int_0^\tau A(u, B, X)Y(u)d\Lambda(u, B^\T X)\bigg\}\bigg\{E \int_0^\tau G^\T(u, B^\T X)du\bigg\}^{-1}.
%\end{gather*}
\end{theorem}

%\section{Algorithm}\label{algorithm}
\section{Numerical Examples}\label{sec:num}
\subsection{Simulation Studies}
In this section, we examine the finite sample performance of our proposed methods via extensive numerical experiments. Specifically, we carry out the estimation of dimension reduction subspace using the forward regression approach \eqref{eq:FR1}, the semiparametric inverse regression approach \eqref{eq:semi}, the counting process inverse regression approach \eqref{eq:dN} and the computational efficient approach \eqref{eq:CPSIR}. All of our methods are implemented through the ``orthoDr'' package in R. Four alternative approaches are considered: a naive approach that performs sliced inverse regression on the failure observations, carried out using the ``dr'' package \citep{weisberg2002dimension}; the double slicing approach \citep{li1999dimension} using R package ``censorSIR'' provided by \cite{wu2008method}; the minimal average variance estimation based on hazard functions in \cite{xia2010dimension}. This implementation is provided by the original author through MATLAB; and the inverse probably of censoring weighted approach based on \cite{lu2011sufficient}. We carry out this approach ourselves by using a Cox proportional hazard model to estimate the censoring weights and obtain the reduced space by utilizing the ``dr'' package with subject weights.

We consider four different settings: Setting $1$ is a classical Cox proportional hazard model; Setting $2$ is constructed with structural dimension $d\!=\!2$ and directions in the hazard function are changing over time; Setting $3$ also has structural dimension equal to two, with the two directions interacting with each other. Setting $4$ also has two interacting structural dimensions, while the failure and the censoring variables have overlap. For each setting, we consider $p = 6$, $12$ and $18$. Each experiment is repeated $200$ times with sample size $n\!=\!400$.

Setting $1$: %We consider the Cox proportional hazard model in this setting.
The true survival time $T$ and the censoring time $C$ are generated from exponential distributions with rate $\exp(\beta^\T X)$ and $\exp(X_4 \!+\! X_5\!-\! 1)$ respectively, where $\beta = (1, 0$$\cdot$$5, 0, \ldots ,0)^\T$ and $X_j$ is the $j$-th element of $X$, for $1\leq j\leq p$. The covariate $X$ follows from multivariate normal distribution with mean $0$ and covariance $\Sigma=\big(0$$\cdot$$5^{|i-j|}\big)_{ij}$. %such that $\Sigma_{ij}$, the $(i,j)$th entry of $\Sigma$, is $0$$.$$5^{|i-j|}$.
The overall censoring rate is around 35$\cdot$3\%.

Setting $2$:  %We set structural dimension $d=2$ and consider a varying coefficient setting.
We generate $T_1$ and $T_2$ from exponential distributions with rate $\exp(\beta_1^\T X)$ and $\exp(\beta_2^\T X)$ respectively, where $\beta_1 = (1, 0, 1, 0, ... ,0)^\T$ and $\beta_2 = (0, 1, 0, 1, 0, \ldots,0)^\T$. The true survival time $T = T_1I(T_1 < 0$$\cdot$$4) + (T_2 + 0$$\cdot$$4)I(T_1 \geq 0$$\cdot$$4)$. The censoring time $C$ is generated from exponential distributions with rate $\exp(X_5 - X_6 - 2)$. The covariate $X$ follows the same distribution as in Setting 1. The overall censoring rate is around 35$\cdot$1\%.

Setting $3$:  %This is another setting with structural dimension $d=2$.
The true survival time $T$ is generated from Weibull distribution with shape parameter $5$ and scale parameter $\exp(4\beta_2^\T X (\beta_1^\T X - 1))$, where $\beta_1 = (1, 0, 1, 0, ... ,0)^\T$ and $\beta_2 = (0, 1, 0, 1, 0, \ldots,0)^\T$. The censoring time $C$ is generated uniformly from 0 to $3 \exp(X_5 - X_6 + 0$$\cdot$$5)$. We further draw $X$ such that $X_j$'s follow standard uniform distribution $U(0,1)$ independently. The overall censoring rate is around 33$\cdot$8\%.

Setting $4$: The true survival time $T$ is generated from a Cox proportional hazard model with $\log(T) = -2$$\cdot$$5 + \beta_1^\T X + 0$$\cdot$$5 \beta_1^\T X \beta_2^\T X + 0$$\cdot$$25 \log(-\log( 1- u))$ and $\log(C) = -0$$\cdot$$5 + \beta_3^\T X + \log(-\log( 1- u))$, where $u$'s are i.i.d. uniform distributed, $\beta_1 = (1, 1, 0, \ldots, 0)^\T$, $\beta_2 = (0, 0, 1, -1, 0, \ldots, 0)^\T$, and $\beta_3 = (0, 1, 0, 1, 1, 1, 0, \ldots, 0)^\T$. The covariate $X$ follows the same distribution as setting 1, except $\Sigma=\big(0$$\cdot$$25^{|i-j|}\big)$. The overall censoring rate is around 26$\cdot$2\%.

\begin{table}[t]
\def~{\hphantom{0}}
\setlength\extrarowheight{-5pt}
\caption{{Simulation results: Mean ($\times 10^2$) and standard deviations ($\times 10^2$, in parenthesis) of Frobenius norm distance (Forb), trace correlation (Tr) and canonical correlation (CCor).}}{\scriptsize
\begin{tabular}{l lllllllll}
Setting 1 (d=1) & \multicolumn{3}{c}{$p=6$} & \multicolumn{3}{c}{$p=12$} & \multicolumn{3}{c}{$p=18$}\\
& \multicolumn{1}{c}{Frob} & \multicolumn{1}{c}{Tr} & \multicolumn{1}{c}{CCor} & \multicolumn{1}{c}{Frob} & \multicolumn{1}{c}{Tr} & \multicolumn{1}{c}{CCor}& \multicolumn{1}{c}{Frob} & \multicolumn{1}{c}{Tr} & \multicolumn{1}{c}{CCor} \\
Naive   &	54	(\!\!	12	\!\!) &	85	~(\!\!	6	\!\!) &	94	~(\!\!	3	\!\!) &	66	(\!\!	12	\!\!) &	78	~(\!\!	8	\!\!) &	92	~(\!\!	3	\!\!) &	73	(\!\!	11	\!\!) &	73	~(\!\!	8	\!\!) &	91	~(\!\!	3	\!\!) \\	
DS      &	33	(\!\!	10	\!\!) &	94	~(\!\!	4	\!\!) &	98	~(\!\!	2	\!\!) &	46	(\!\!	11	\!\!) &	89	~(\!\!	5	\!\!) &	97	~(\!\!	2	\!\!) &	53	(\!\!	10	\!\!) &	85	~(\!\!	5	\!\!) &	96	~(\!\!	2	\!\!) \\	
IPCW-SIR&	64	(\!\!	13	\!\!) &	78	~(\!\!	9	\!\!) &	91	~(\!\!	4	\!\!) &	75	(\!\!	11	\!\!) &	71	~(\!\!	9	\!\!) &	89	~(\!\!	4	\!\!) &	80	(\!\!	11	\!\!) &	68	~(\!\!	9	\!\!) &	89	~(\!\!	4	\!\!) \\	
hMave	&	68	(\!\!	12	\!\!) &	76	~(\!\!	8	\!\!) &	86	~(\!\!	5	\!\!) &	73	(\!\!	11	\!\!) &	73	~(\!\!	8	\!\!) &	86	~(\!\!	5	\!\!) &	79	(\!\!	10	\!\!) &	68	~(\!\!	8	\!\!) &	84	~(\!\!	5	\!\!) \\	
Forward	&	21	~(\!\!	6	\!\!) &	98	~(\!\!	1	\!\!) &	99	~(\!\!	0	\!\!) &	33	~(\!\!	8	\!\!) &	94	~(\!\!	3	\!\!) &	99	~(\!\!	1	\!\!) &	39	~(\!\!	7	\!\!) &	92	~(\!\!	3	\!\!) &	98	~(\!\!	1	\!\!) \\	
CP-SIR	&	26	~(\!\!	9	\!\!) &	96	~(\!\!	3	\!\!) &	99	~(\!\!	1	\!\!) &	40	(\!\!	10	\!\!) &	91	~(\!\!	4	\!\!) &	98	~(\!\!	1	\!\!) &	49	~(\!\!	9	\!\!) &	88	~(\!\!	4	\!\!) &	97	~(\!\!	1	\!\!) \\	
IR-CP	&	23	~(\!\!	7	\!\!) &	97	~(\!\!	2	\!\!) &	99	~(\!\!	0	\!\!) &	35	~(\!\!	8	\!\!) &	94	~(\!\!	3	\!\!) &	98	~(\!\!	1	\!\!) &	41	~(\!\!	7	\!\!) &	91	~(\!\!	3	\!\!) &	98	~(\!\!	1	\!\!) \\	
IR-Semi	&	23	~(\!\!	8	\!\!) &	97	~(\!\!	2	\!\!) &	99	~(\!\!	0	\!\!) &	37	~(\!\!	8	\!\!) &	93	~(\!\!	3	\!\!) &	98	~(\!\!	1	\!\!) &	44	~(\!\!	8	\!\!) &	90	~(\!\!	4	\!\!) &	98	~(\!\!	1	\!\!) \\	
{\smallskip\vspace{-0.1in}}\\
Setting 2 (d=2) & \multicolumn{3}{c}{$p=6$} & \multicolumn{3}{c}{$p=12$} & \multicolumn{3}{c}{$p=18$}\\
& \multicolumn{1}{c}{Frob} & \multicolumn{1}{c}{Tr} & \multicolumn{1}{c}{CCor} & \multicolumn{1}{c}{Frob} & \multicolumn{1}{c}{Tr} & \multicolumn{1}{c}{CCor}& \multicolumn{1}{c}{Frob} & \multicolumn{1}{c}{Tr} & \multicolumn{1}{c}{CCor} \\
Naive	&	67	(\!\!	19	\!\!) &	88	~(\!\!	6	\!\!) &	96	~(\!\!	3	\!\!) &	87	(\!\!	18	\!\!) &	80	~(\!\!	8	\!\!) &	91	~(\!\!	6	\!\!) &	\!\!\!106	(\!\!	17	\!\!) &	71	~(\!\!	9	\!\!) &	87	~(\!\!	6	\!\!) \\
DS      &	44	(\!\!	13	\!\!) &	95	~(\!\!	3	\!\!) &	98	~(\!\!	1	\!\!) &	68	(\!\!	15	\!\!) &	88	~(\!\!	5	\!\!) &	94	~(\!\!	3	\!\!) &	84	(\!\!	11	\!\!) &	82	~(\!\!	5	\!\!) &	92	~(\!\!	3	\!\!) \\
IPCW-SIR&	83	(\!\!	19	\!\!) &	82	~(\!\!	8	\!\!) &	94	~(\!\!	3	\!\!) &	98	(\!\!	17	\!\!) &	75	~(\!\!	9	\!\!) &	90	~(\!\!	6	\!\!) &	\!\!\!114	(\!\!	16	\!\!) &	67	~(\!\!	9	\!\!) &	86	~(\!\!	8	\!\!) \\
hMave	&	\!\!\!114	(\!\!	31	\!\!) &	65	(\!\!	16	\!\!) &	74	(\!\!	16	\!\!) &	\!\!\!139	(\!\!	19	\!\!) &	51	(\!\!	12	\!\!) &	64	(\!\!	12	\!\!) &	\!\!\!151	(\!\!	14	\!\!) &	43	(\!\!	10	\!\!) &	59	(\!\!	10	\!\!) \\
Forward	&	\!\!\!102	~(\!\!	1	\!\!) &	49	~(\!\!	1	\!\!) &	\!\!\!100	~(\!\!	0	\!\!) &	\!\!\!105	~(\!\!	2	\!\!) &	48	~(\!\!	1	\!\!) &	99	~(\!\!	0	\!\!) &	\!\!\!107	~(\!\!	2	\!\!) &	46	~(\!\!	1	\!\!) &	99	~(\!\!	0	\!\!) \\
CP-SIR	&	37	(\!\!	11	\!\!) &	96	~(\!\!	2	\!\!) &	98	~(\!\!	1	\!\!) &	61	(\!\!	12	\!\!) &	90	~(\!\!	4	\!\!) &	96	~(\!\!	2	\!\!) &	78	(\!\!	10	\!\!) &	85	~(\!\!	4	\!\!) &	93	~(\!\!	2	\!\!) \\
IR-CP	&	49	(\!\!	19	\!\!) &	93	~(\!\!	6	\!\!) &	96	~(\!\!	3	\!\!) &	73	(\!\!	20	\!\!) &	86	~(\!\!	8	\!\!) &	92	~(\!\!	4	\!\!) &	90	(\!\!	17	\!\!) &	79	~(\!\!	8	\!\!) &	89	~(\!\!	5	\!\!) \\
IR-Semi	&	39	(\!\!	14	\!\!) &	96	~(\!\!	3	\!\!) &	98	~(\!\!	2	\!\!) &	65	(\!\!	16	\!\!) &	89	~(\!\!	6	\!\!) &	94	~(\!\!	3	\!\!) &	83	(\!\!	15	\!\!) &	82	~(\!\!	6	\!\!) &	91	~(\!\!	3	\!\!) \\
{\smallskip\vspace{-0.1in}}\\
Setting 3 (d=2) & \multicolumn{3}{c}{$p=6$} & \multicolumn{3}{c}{$p=12$} & \multicolumn{3}{c}{$p=18$}\\
& \multicolumn{1}{c}{Frob} & \multicolumn{1}{c}{Tr} & \multicolumn{1}{c}{CCor} & \multicolumn{1}{c}{Frob} & \multicolumn{1}{c}{Tr} & \multicolumn{1}{c}{CCor}& \multicolumn{1}{c}{Frob} & \multicolumn{1}{c}{Tr} & \multicolumn{1}{c}{CCor} \\
Naive	&	72	(\!\!	23	\!\!) &	86	~(\!\!	9	\!\!) &	96	~(\!\!	5	\!\!) &	99	(\!\!	22	\!\!) &	74	(\!\!	11	\!\!) &	88	(\!\!	11	\!\!) &	\!\!\!116	(\!\!	18	\!\!) &	66	(\!\!	10	\!\!) &	82	(\!\!	13	\!\!) \\	
DS      &	40	(\!\!	14	\!\!) &	95	~(\!\!	3	\!\!) &	99	~(\!\!	1	\!\!) &	60	(\!\!	13	\!\!) &	91	~(\!\!	4	\!\!) &	97	~(\!\!	3	\!\!) &	73	(\!\!	15	\!\!) &	86	~(\!\!	6	\!\!) &	95	~(\!\!	5	\!\!) \\	
IPCW-SIR&	\!\!\!113	(\!\!	26	\!\!) &	66	(\!\!	13	\!\!) &	81	(\!\!	14	\!\!) &	\!\!\!129	(\!\!	15	\!\!) &	58	~(\!\!	9	\!\!) &	74	(\!\!	12	\!\!) &	\!\!\!133	(\!\!	11	\!\!) &	55	~(\!\!	7	\!\!) &	74	(\!\!	12	\!\!) \\	
hMave	&	40	(\!\!	18	\!\!) &	95	~(\!\!	6	\!\!) &	99	~(\!\!	3	\!\!) &	66	(\!\!	27	\!\!) &	87	(\!\!	12	\!\!) &	94	~(\!\!	9	\!\!) &	89	(\!\!	29	\!\!) &	78	(\!\!	14	\!\!) &	89	(\!\!	12	\!\!) \\	
Forward	&	\!\!\!100	~(\!\!	0	\!\!) &	50	~(\!\!	0	\!\!) &	\!\!\!100	~(\!\!	0	\!\!) &	\!\!\!100	~(\!\!	0	\!\!) &	50	~(\!\!	0	\!\!) &	\!\!\!100	~(\!\!	0	\!\!) &	\!\!\!101	~(\!\!	0	\!\!) &	50	~(\!\!	0	\!\!) &	\!\!\!100	~(\!\!	0	\!\!) \\	
CP-SIR	&	34	(\!\!	11	\!\!) &	97	~(\!\!	2	\!\!) &	99	~(\!\!	1	\!\!) &	55	(\!\!	11	\!\!) &	92	~(\!\!	3	\!\!) &	97	~(\!\!	2	\!\!) &	67	(\!\!	11	\!\!) &	88	~(\!\!	4	\!\!) &	96	~(\!\!	3	\!\!) \\	
IR-CP	&	30	(\!\!	14	\!\!) &	97	~(\!\!	3	\!\!) &	99	~(\!\!	1	\!\!) &	46	(\!\!	14	\!\!) &	94	~(\!\!	4	\!\!) &	99	~(\!\!	1	\!\!) &	58	(\!\!	15	\!\!) &	91	~(\!\!	5	\!\!) &	97	~(\!\!	4	\!\!) \\	
IR-Semi	&	19	~(\!\!	8	\!\!) &	99	~(\!\!	1	\!\!) &	\!\!\!100	~(\!\!	0	\!\!) &	29	~(\!\!	8	\!\!) &	98	~(\!\!	1	\!\!) &	\!\!\!100	~(\!\!	0	\!\!) &	40	(\!\!	11	\!\!) &	96	~(\!\!	2	\!\!) &	99	~(\!\!	1	\!\!) \\	
{\smallskip\vspace{-0.1in}}\\
Setting 4 (d=2) & \multicolumn{3}{c}{$p=6$} & \multicolumn{3}{c}{$p=12$} & \multicolumn{3}{c}{$p=18$}\\
& \multicolumn{1}{c}{Frob} & \multicolumn{1}{c}{Tr} & \multicolumn{1}{c}{CCor} & \multicolumn{1}{c}{Frob} & \multicolumn{1}{c}{Tr} & \multicolumn{1}{c}{CCor}& \multicolumn{1}{c}{Frob} & \multicolumn{1}{c}{Tr} & \multicolumn{1}{c}{CCor} \\
Naive	&	33	~(\!\!	9	\!\!) &	97	~(\!\!	2	\!\!) &	99	~(\!\!	1	\!\!) &	52	(\!\!	10	\!\!) &	93	~(\!\!	3	\!\!) &	97	~(\!\!	1	\!\!) &	66	(\!\!	10	\!\!) &	89	~(\!\!	4	\!\!) &	95	~(\!\!	2	\!\!) \\
DS      &	49	(\!\!	12	\!\!) &	94	~(\!\!	3	\!\!) &	95	~(\!\!	3	\!\!) &	62	(\!\!	11	\!\!) &	90	~(\!\!	4	\!\!) &	94	~(\!\!	3	\!\!) &	71	(\!\!	11	\!\!) &	87	~(\!\!	4	\!\!) &	93	~(\!\!	3	\!\!) \\
IPCW-SIR&	35	~(\!\!	9	\!\!) &	97	~(\!\!	2	\!\!) &	98	~(\!\!	1	\!\!) &	52	(\!\!	10	\!\!) &	93	~(\!\!	3	\!\!) &	97	~(\!\!	1	\!\!) &	64	(\!\!	10	\!\!) &	89	~(\!\!	3	\!\!) &	95	~(\!\!	2	\!\!) \\
hMave	&	\!\!\!142	~(\!\!	3	\!\!) &	50	~(\!\!	2	\!\!) &	59	~(\!\!	4	\!\!) &	\!\!\!145	~(\!\!	5	\!\!) &	47	~(\!\!	4	\!\!) &	57	~(\!\!	4	\!\!) &	\!\!\!149	~(\!\!	7	\!\!) &	45	~(\!\!	6	\!\!) &	55	~(\!\!	6	\!\!) \\
Forward	&	\!\!\!101	~(\!\!	0	\!\!) &	50	~(\!\!	0	\!\!) &	\!\!\!100	~(\!\!	0	\!\!) &	\!\!\!102	~(\!\!	1	\!\!) &	49	~(\!\!	0	\!\!) &	99	~(\!\!	0	\!\!) &	\!\!\!102	~(\!\!	1	\!\!) &	49	~(\!\!	0	\!\!) &	99	~(\!\!	0	\!\!) \\
CP-SIR	&	36	~(\!\!	7	\!\!) &	97	~(\!\!	1	\!\!) &	98	~(\!\!	1	\!\!) &	51	~(\!\!	8	\!\!) &	93	~(\!\!	2	\!\!) &	97	~(\!\!	1	\!\!) &	63	~(\!\!	8	\!\!) &	90	~(\!\!	3	\!\!) &	95	~(\!\!	1	\!\!) \\
IR-CP	&	22	~(\!\!	9	\!\!) &	99	~(\!\!	2	\!\!) &	99	~(\!\!	1	\!\!) &	42	(\!\!	15	\!\!) &	95	~(\!\!	4	\!\!) &	98	~(\!\!	2	\!\!) &	57	(\!\!	17	\!\!) &	91	~(\!\!	5	\!\!) &	96	~(\!\!	3	\!\!) \\
IR-Semi	&	13	~(\!\!	5	\!\!) &	99	~(\!\!	0	\!\!) &	\!\!\!100	~(\!\!	0	\!\!) &	24	~(\!\!	7	\!\!) &	98	~(\!\!	1	\!\!) &	99	~(\!\!	0	\!\!) &	34	(\!\!	10	\!\!) &	97	~(\!\!	2	\!\!) &	99	~(\!\!	1	\!\!)
\end{tabular}}\label{tab:sim}
{\smallskip\vspace{0.1in}}\\
\scriptsize DS: \cite{li1999dimension}; IPCW-SIR: \cite{lu2011sufficient}; hMave: \cite{xia2010dimension}; Forward: forward regression; CP-SIR: the computational efficient approach; IR-CP: the counting process inverse regression approach; IR-Semi: the semiparametric inverse regression approach.
\end{table}

We first investigate the statistical performance using three different measures: the Frobenius norm distance between the projection matrix $P$ and its estimator $\widehat P$, where $P = B(B^\T B)^{-1} B^\T$; the trace correlation $\textnormal{tr}\big(P \widehat P\big)/d$, where $d$ is the structural dimension; and the canonical correlation between $B^\T X$ and ${\widehat B}^\T X$. The results are summarized in Table \ref{tab:sim}.

Overall, the two inverse regression methods achieve the best performance, followed by the computationally efficient approach. It is worth to point out that the computational efficient approach, while no non-parametric approximation is required, outperforms existing methods in almost all settings. Among all competing methods, double slicing performs the best in general, while \cite{xia2010dimension} and \cite{lu2011sufficient} outperforms double slicing in Setting 3 and Setting 4, respectively. In terms of the three error measurements, we found that the Frobenius norm distance is the most informative measurement, while the Trace and Canonical correlations are less sensitive to the performances.

Among the two inverse regression methods, the semiparametric version is slightly better in Settings 3 and 4. The main advantage of the semiparametric version compared with the counting process version is the double robustness, which ensures consistency even when the conditional expectations are not estimated correctly. However, this theoretical advantage does not translate into strong numerical improvements in Settings 1 and 2 especially when $p$ is large. This is possibly due to the variations in the hazard function estimation, which introduces less numerical stability. In setting 1, forward regression approach achieves the best performance. As discussed in Example 1, this method mimics the efficient estimating equations used in the Cox proportional hazard model and is thus the most efficient method in this setting. In setting 2, The computationally efficient approach performs similarly to the two inverse regression approaches and even outperforms them under large $p$. This shows some potential of this approach in higher dimension settings when nonparametric estimations may not be preferred.

One major challenge of solving estimating equations is the computational complexity, especially with nonparametric components. Our proposed method adds additional difficulties with the orthogonality constraints, i.e., $B^\T B = I_{d \times d}$. However, with our proposed orthogonality constrained optimization algorithm, in combination with the Rcpp interface, our implementation can solve the proposed method very efficiently. In addition, parallel computing through OpenMP is utilized to numerically approximate the gradient for each entry of $B$. For example, in Setting 2 with $p=6$, the mean computational time of the inverse regression counting process approach is 1.62 seconds, while the time for the semiparametric version is 8.01 seconds. Table \ref{tab:time} in the Supplementary Material summarizes the computational cost across all settings. All simulations are done on an Intel Xeon E5-2680v4 processor with 5 parallel threads.

We further investigate the variance of the proposed methods. Due to the complicated form of the variance formula, we instead use the bootstrap to obtain an estimation of the standard deviation of the proposed estimators. Using an upper-block-diagonal version of the parameter of interest, we estimate the standard deviation of the parameters based on 100 bootstrap samples and also report the 95\% confidence interval. The results show that in setting one, the bootstrap estimator of all the proposed methods approximates the standard deviation closely. In the rest settings, the approximation of the computational efficient approach and counting process inverse regression approach still archives good performance, while semiparametric inverse regression approach slightly over-estimate the standard deviations, leads to a slight over-coverage (around 98\%). However, the proposed methods still archive smaller empirical standard deviation on nearly all parameters across all settings. Details are provided in the Supplementary Material.

\subsection{Skin Cutaneous Melanoma Data Analysis}\label{realdata}\label{sec:6}
We apply the proposed method to The Cancer Genome Atlas (TCGA, \url{http://cancergenome.nih.gov/}) skin cutaneous melanoma dataset. TCGA provides the public with one of the most comprehensive profiling data on more than thirty cancer types. We acquire gene expression and clinical data on a total of 469 patients (156 observed failures) and their mRNA expression data on 20,531 genes. To produce biologically meaningful results, we preselect 20 genes in this analysis, which are the top 20 genes highly associated with cutaneous melanoma based on meta-analyses of over 145 existing literature \citep{chatzinasiou2011comprehensive}. A list of these genes can be found at \url{http://bioserver-3.bioacademy.gr/Bioserver/MelGene/}. We further include age at diagnosis as a clinical control variable. All covariates are pre-processed to have unit variance and zero mean.

Selecting the number of structural dimensions can be a challenging task, especially with right censored survival model \citep{xia2010dimension}. To this end, we adopt the validated information criterion developed by \citet{ma2015validated}, which is particularly suited for our generalized method of moments framework. The validated information criterion is constructed by penalizing the quadratic form of the objective function. Interestingly when we apply this method to all of our proposed estimating equation approaches, $d=1$ always yields the best fit. Hence we present the results for all method under $d=1$. As a demonstration of the fitted model, we project the design matrix on the estimated direction of the semiparametric inverse regression approach and plot the survival outcome against the projection (Figure \ref{fig:skcm.semi}). A nonparametric estimation of the conditional survival function based on this projection is also produced. From these two figures, we can see a clear trend that subjects with larger values of the projection have lower survival rate. As for comparisons, we perform the competing methods with one structural dimension, and the results can be found in the Supplementary Material (Figure \ref{fig:skcm.alt}). From this simple visualization, it seems that the double slicing method obtains a similar direction with monotone effects on the risk of failure, while the other directions obtained by other methods are non-monotone.

We also observe both similarities and differences among different methods for the identified genes. A pairwise Frobenius norm distance is provided in the Supplementary Material (Table \ref{SKCM:dist}). This suggests that the proposed methods have fairly small distances while existing methods mostly do not agree with each other. In addition, double slicing has the smallest distance with the proposed methods. The proposed methods share some consistent trend of the loadings on the most influential variables. For example, all the proposed methods identify Age as the most important variable with loadings over 0.5. Corresponding this with the survival curve plot in Figure \ref{fig:skcm.semi}, it suggests that higher survival rate is observed for smaller derived direction. This further indicates that patients with younger age tend to have higher survival, which is biologically intuitive. This finding is also consistent with the double slicing method, which identifies Age with loading 0$\cdot$47. However, other methods do not assign large loading to Age. Here, we present all results with a positive loading of age, and the directions are multiplied by $-1$ otherwise. Another important variable that all methods agree in terms of signs is MTAP. This gene has been previously reported to have a negative correlation with the progression of melanocytic tumors \cite{behrmann2003characterization}, which justifies the large negative value estimated by the proposed methods. However, the magnitudes in alternative methods are small. Other common genes identified by the proposed methods are MYH7B and CASP8. We mine the literature and found that \cite{li2008genetic} genotyped putatively functional polymorphisms of CASP8 and found a significant association with lower risk of cutaneous melanoma. The result therein supports the large negative loading of CASP8 gene in our fitted model. For differences across methods, TYR is identified by alternative methods except the double slicing, with extremely large loadings (up to $-$0$\cdot$78). The enzyme encoded by this gene controls the production of melanin, hence it has been shown to be strongly associated with melanoma \citep{gudbjartsson2008asip}. Although the estimated directions are dominated by this gene, we did not observe a monotone effect of the directions (See Figure \ref{fig:skcm.alt} in the Supplementary Material).

\begin{table}[htbp]
    \renewcommand{\arraystretch}{0.75}
    \caption{SKCM data analysis results: the loading vectors ($\times 10^2$) of the first structural dimension.} 
    \center{\scriptsize
    \begin{tabular}{r r r r r r r r r}
 & \begin{turn}{90} \begin{tabular}{l} Naive \end{tabular} \end{turn} & \begin{turn}{90} \begin{tabular}{l} DS \end{tabular} \end{turn} &
 \begin{turn}{90} \begin{tabular}{l} IPCW-SIR \end{tabular} \end{turn} & \begin{turn}{90} \begin{tabular}{l} hMave \end{tabular} \end{turn} &
 \begin{turn}{90} \begin{tabular}{l} \\ Forward \end{tabular} \end{turn} & \begin{turn}{90} \begin{tabular}{l} CP-SIR \end{tabular} \end{turn} &
 \begin{turn}{90} \begin{tabular}{l} IR-CP \end{tabular} \end{turn} & \begin{turn}{90} \begin{tabular}{l} IR-Semi \end{tabular} \end{turn} \\
Age	&	16	&	47	&	10	&	0	&	60	&	59	&	53	&	54	\\
TYRP1	&	-16	&	-5	&	-9	&	24	&	18	&	11	&	39	&	30	\\
OCA2	&	18	&	17	&	14	&	-6	&	21	&	19	&	22	&	5	\\
TYR	&	-60	&	-9	&	-65	&	-78	&	-19	&	-27	&	-19	&	9	\\
SLC45A2	&	11	&	24	&	23	&	14	&	30	&	28	&	16	&	17	\\
CDKN2A	&	6	&	-28	&	-2	&	-12	&	-9	&	-7	&	-2	&	-11	\\
MX2	&	2	&	-2	&	-2	&	-12	&	-19	&	-13	&	-30	&	-27	\\
MTAP	&	-15	&	-8	&	-10	&	-14	&	-31	&	-36	&	-35	&	-30	\\
MITF	&	56	&	-9	&	43	&	5	&	-13	&	-12	&	2	&	-27	\\
VDR	&	5	&	-18	&	-9	&	10	&	-10	&	-6	&	-4	&	2	\\
CCND1	&	-20	&	35	&	-21	&	-5	&	16	&	17	&	18	&	16	\\
MYH7B	&	10	&	-27	&	5	&	-4	&	-29	&	-32	&	-30	&	-48	\\
ATM	&	-16	&	-22	&	2	&	28	&	-4	&	7	&	0	&	6	\\
PLA2G6	&	-22	&	-16	&	-21	&	7	&	4	&	-5	&	-11	&	-3	\\
CASP8	&	15	&	-39	&	21	&	-13	&	-26	&	-24	&	-18	&	-14	\\
AFG3L1	&	12	&	26	&	18	&	-15	&	17	&	10	&	-6	&	-9	\\
CDK10	&	3	&	8	&	2	&	25	&	-7	&	-1	&	9	&	8	\\
PARP1	&	-9	&	3	&	-22	&	17	&	14	&	18	&	8	&	18	\\
CLPTM1L	&	-8	&	-5	&	17	&	2	&	-6	&	-6	&	-2	&	-6	\\
ERCC5	&	-14	&	25	&	-17	&	-7	&	12	&	13	&	22	&	3	\\
FTO	&	-3	&	-3	&	-8	&	14	&	15	&	17	&	7	&	5	\\
\end{tabular}}\label{SKCM}
{\smallskip\vspace{0.1in}}\\
\scriptsize DS: \cite{li1999dimension}; IPCW-SIR: \cite{lu2011sufficient}; hMave: \cite{xia2010dimension}; Forward: forward regression; CP-SIR: the computational efficient approach; IR-CP: the counting process inverse regression approach; IR-Semi: the semiparametric inverse regression approach.
\end{table}

\begin{figure}[t]
\centering
\caption{Fitted direction and survival function of the semiparametric inverse regression}{
    \begin{tabular}{cc}
    \includegraphics[height = 2.5in, width=0.45\textwidth, trim={5.5in 10.5in 4in 5in}, clip]{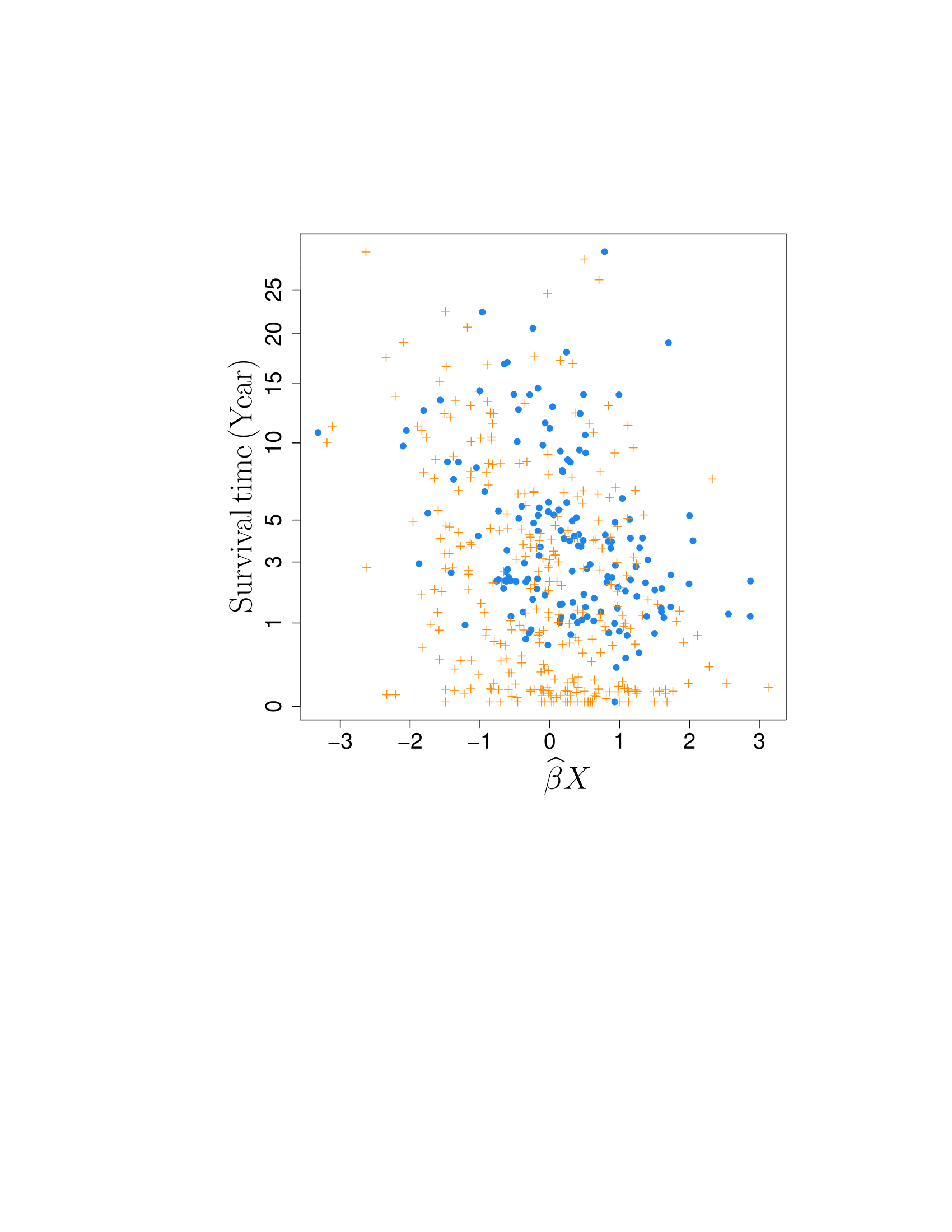} \hspace{0.05\textwidth}
    \includegraphics[height = 2.5in, width=0.45\textwidth, trim={5.5in 11in 3.5in 6in}, clip]{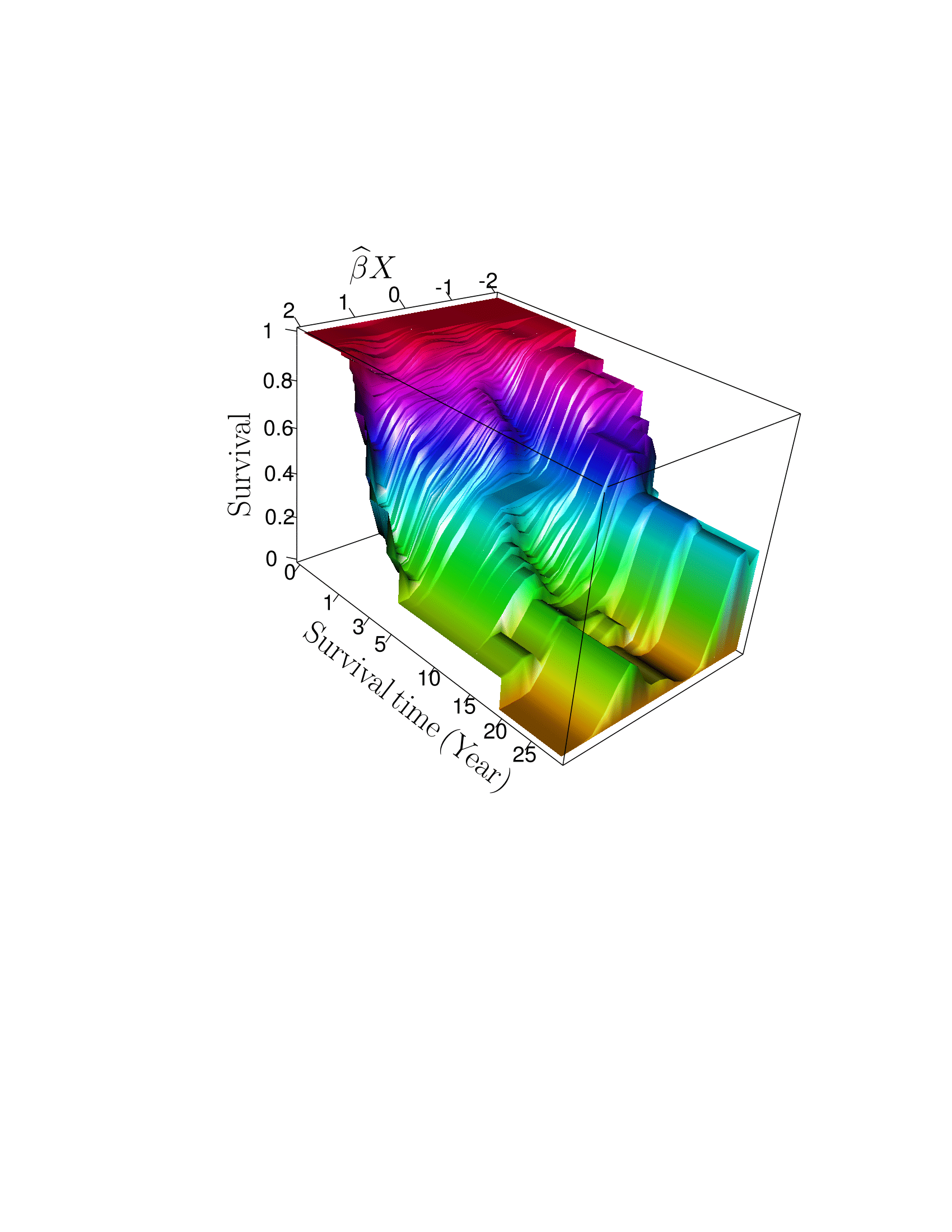}
    \end{tabular}}\label{fig:skcm.semi}
    \scriptsize The left figure is the projected direction versus the observed failure (blue dot) and censoring (orange $+$) times. The right figure is a nonparametric estimation of the survival function based on the projected direction.
\end{figure}

\section{Discussion}\label{sec:dis}

In this paper, we proposed a counting process based dimension reduction framework for censored outcomes. A family of generalized method of moments based approaches has been constructed for estimating the dimension reduction subspace. The main advantage of the proposed method is that it requires only a $d$-dimensional (instead of $p$) nonparametric kernel estimation while no censoring distribution is modeled. The reduced dimension of the nonparametric estimation circumvents the difficulties of many existing methods and improves the efficiency when the total dimension $p$ is too large for kernel methods. Our simulation study suggests that the proposed method outperforms existing methods in a variety of settings. To efficiently solve the proposed estimating equations, we further introduce an orthogonality constrained optimization method that solves the parameters within a Stiefel manifold. With implementations in the R package ``orthoDr'' through C++, the counting process version of the estimators can be solved within a few seconds. However, the martingale version requires significant more calculations due to the local estimation the hazard function, hence requiring a few minutes to solve. We believe that there is still potential room to improve the computational performance. In addition, our computational efficient approach requires only a singular value decomposition and has satisfactory performances. However, it does not enjoy the same theoretical guarantee without restraint conditions on the covariates. Further relaxation of these conditions is of great interest.

Our framework can be possibly extended to more general settings. First, by imposing penalization on the estimating equations, it is possible to extend the proposed method to moderately high dimensional data. Sparse estimation of the $B$ parameter may help both interpretations and improves the prediction accuracy of subsequent nonparametric models. The second direction is to search for an alternative construction of the $\alpha$ functions. Throughout our developments, we used the $\phi(u)$ function which is motivated by the inverse regression of a Bernoulli distribution. It would be interesting to investigate the possibilities of a ``SAVE'' type of $\alpha$ function that many detect more complicated model structure. We can also consider using $\alpha(u, X) = B^\T X \varphi^\T(u)$, which is another valid choice. Lastly, it is also interesting to extend this framework to a time-varying coefficient setting, where we may let the dimension reduction space $\cal S$ to change over time $t$.

\section{Supplementary Material}
Supplementary material available at Biometrika online includes a derivation of the ortho-complement of the nuisance tangent space; the proof of the double robustness property for the semiparametric inverse regression approach; the proof of Lemma \ref{lem:thm.1} and Theorem \ref{thm1}; and additional simulation and data analysis results.

\section{Acknowledgement}

The first two authors contribute equally. This work is partially supported by the NIH and the NSERC.

\newpage

\begin{center}
\Huge{\textbf{Supplementary Material}}
\end{center}

\appendix

\section{Tangent Space Derivation}\label{supp:tangent}

Before we give the tangent space, we first derive the nuisance tangent spaces $\cE_1, \cE_2$ and $\cE_3$ in the following proposition. The proof follows the similar argument for proving the nuisance tangent spaces of cox regression model \citep{tsiatis2007semiparametric} and is thus  omitted for simplicity.
\begin{proposition}\label{prop:nuisance}
The nuisance tagent spaces $\cE_1, \cE_2$ and $\cE_3$ have the following forms
\$
\cE_1&=\bigg\{ \int \alpha(u, B^\T X) dM(u,B^\T X):\, \alpha(u,B^\T X) \,\, \text{is measurable} \bigg\},\\
\cE_2&=\bigg\{ \int \alpha(u, X) dM_C(u,X):\, \alpha(u,X) \,\, \text{is measurable} \bigg\}, ~\text{and}\\
\cE_3&=\bigg\{  \alpha(X): E \big\{\alpha(X)\big\}=0\bigg\},\$
\end{proposition}
Next, we give a proof of tangent space $\cE^\perp$, which is defined as
\#
\cE^\perp=\bigg\{ \int \big\{\alpha(u, X)-\alpha^*(u, B^\T X)\big\} dM(u,X): \, \alpha(u,X) \,\, \text{is measurable} \bigg\}. \nonumber
\#
where
\#
\alpha^*(u,B^\T X) = E\big\{\alpha(u,X)\big| \cF_{u}, B^\T X\big\}.\nonumber
\#
\begin{proof}[Proof of Proposition \ref{prop:nuisance}]
For a fully nonparametric model, the nuisance tangent space is the whole Hilbert space $\cH$ with each element having  mean zero. Therefore, if we put no restriction on the hazard function $\lambda({t|X})$ and write the associated nuisance tangent space as $\cE^*_1$, we obtain
\begin{gather*}
\cH=\cE^*_1\oplus \cE_2\oplus\cE_3, ~\text{where}~\\
\cE^*_1= \bigg\{ \int \alpha(u, X)dM(u,X):\, \alpha(u,X) \,\, \text{is measurable} \bigg\}.
\end{gather*}
The orthogonal completion  of $\cE$ satisfies that $\cE^\perp\subset \cE_1^*$  and $\cE^\perp \perp \cE_1$. In order to identify $\cE^\perp$, it suffices to take an arbitrary element in $\cE_1^*$ and find its residual after projecting it onto $\cE_1$. To find the projection, we must derive $\alpha^*(u, B^\T X)\in \cE_1$ such that
\$
E\bigg(  \int\big(\alpha(u,X)-\alpha^*(u,B^\T X)\big)^\T dM(u,X)\int a(u,B^\T X) dM(u,X)  \bigg)=0.
\$
The covariance of martingale stochastic integrals  above can be computed by finding the expectation of the predictable covariance process \citep{fleming2011counting}:
%\$
%&E\bigg(  \int\big\{\alpha(u,X)-\alpha^*(u,B^\T X)\big\}^\T dM(u,X)\int a(u,B^\T X) dM(u,X)  \bigg)\\
%&=E\bigg( \int  E \Big[   \alpha(u,X)Y(u) -\alpha^*(u,B^\T X)Y(u) \big|B^\T X \Big]^\T a(u,B^\T X) \lambda(u, B^\T X) du\bigg)=0
%\$
%for all $a(u,B^\T X)\in \cE_1$.
%Since $a(u,B^\T X)$ is arbitrary, thus we must have
%\$
%E\big\{  \alpha^*(u,B^\T X)Y(u)\big| B^\T X\big  \}
%=E\big\{\alpha (u,X)Y(u)\big| B^\T X\big\},
%\$
%which further yields that
%\$
%\alpha^*(u,B^\T X)=\frac{E\{\alpha(u,X)Y(u)|B^\T X\}}{E\{Y(u)|B^\T X\}}=E\big\{\alpha(u,X)|Y(u)=1, B^\T X \big\}.
%\$
%{\bf Other Derivation}:\\
%On the other hand, if we condition on $\{\cF_{u-}, B^\T X\}$, we obtain that
\$
&E\bigg(  \int\big\{\alpha(u,X)-\alpha^*(u,B^\T X)\big\}^\T dM(u,X)\int a(u,B^\T X) dM(u,X)  \bigg)\\
&=E\bigg( \int   \Big[E\big\{   \alpha(u,X)\big|\cF_{u}, B^\T X\big\}-\alpha^*(u,B^\T X) \Big]^\T a(u,B^\T X) \lambda(u|B^\T X)Y(u) du\bigg)=0,
\$
where $a(u, B^\T X)$ is arbitrary and $\cF_{u}$ is the filtration. Thus we must have
\$
\alpha^*(u, B^\T X)= E\big\{\alpha(u,X)\big|\cF_{u}, B^\T X\big\}.
\$
This completes the proof. %When $\alpha(u, X)$ does not depend on $\cF_{u}$, $\alpha^*(u, B^\T X)$ can be further written as
%\$
%\alpha^*(u,B^\T X)=\frac{E\{\alpha(u,X)Y(u)|B^\T X\}}{E\{Y(u)|B^\T X\}}=E\big\{\alpha(u,X)|Y(u)=1, B^\T X \big\}.
%\$
\end{proof}

\section{Proof of the Double Robustness Property for the semiparametric inverse regression approach}\label{supp:double}
Recall that, for the semiparametric inverse regression approach, we solve the sample version of the following estimating functions
\$
E\bigg[\int \Big\{E\big(X|Y(u)\big)-E\big(X| Y(u), B^\T X\big)\Big\}\varphi^\T(u)dM(u)\bigg]=0.
\$
For simplicity, we will use the random function to denote
\$
F(X, u)= \Big\{E\big(X|Y(u)\big)-E\big(X|Y(u), B^\T X\big)\Big\} \varphi^\T(u).
\$

\noindent {\sf Case 1}: Suppose $M(u)$ is misspecified as $M^*(u)$. Then we have that
\$
E\big\{F(X, u)d M^*(u)\big\}&\!=\!E\big[E\big\{F(X, u) dM^*(u)| Y(u), X\big\}\big]\\
&\!=\!E\big[F(X, u)E\big\{dM^*(u)| Y(u), X\big\}\big]\\
&\!=\!E\big[E\big\{F(X, u)| Y(u), B^\T X\big\}E\big\{dM^*(u)| Y(u), B^\T X\big\}\big]\\
&\!=\!0,
\$
where the last equation is due to the fact that $E\big\{F(X, u)| Y(u), B^\T X\big\} = 0$. Hence we have
\$
E\bigg[\int F(X,u)dM^*(u)\bigg]=0.
\$

\noindent {\sf Case 2}: Suppose the function $F(X,u)$ is misspecified to $F^*(X,u)$. With a similar argument, we can show that $E\big\{\int F^*(X,u )dM(u)\big\}=0$ due to the fact that $E\big\{dM(u)| Y(u), B^\T X\big\} = 0$. This completes the proof.

\section{Proofs of asymptotic results}

\subsection{Proof of Lemma \ref{lem:thm.1}}
\begin{proof}[Proof of Lemma \ref{lem:thm.1}]
Without loss of generality, we prove the results for $d=1$.
We start with  the convergence rate of $\widehat\lambda(t|B^\T X=B^\T x)$. Let $f_{B^\T X}(z)$ be the true density  function of $B^\T X$ evaluated at $z=B^\T x$, and
\$
A_i=\frac{1}{n}\sum_{j=1}^n I(Y_j\geq u)K_h(B^\T X_j-z)-f_{B^\T X}(z)E\big\{I(Y\geq u)| z \big\}.
\$
Suppose we have proved that $A_i=O_p\left((nh)^{-1/2}+h^2\right)$ uniformly.
Let $z_i=B^\T x_i$.
We have
\$
\widehat\Lambda(t| z)&=\sum_{i=1}^n\frac{ K_b(Y_i-t)\delta_iK_h(B^\T X_i-z) }{\sum_{j=1}^n I(Y_j\geq t)K_h(B^\T X_i-z)}\\
&=\underbrace{\frac{1}{n}\sum_{i=1}^n\frac{ K_b(Y_i-t)\delta_iK_h(B^\T X_i-z) }{f_{B^\T X}(z)E\big\{I(Y\geq t)| z \big\}}}_{\Rom{1}}\big(1+O(A_i)\big).
\$
We bound the term (\Rom{1}) first. Let $S(t, x)=\pr(T\geq t| X=x)$ and $S_c(t, x)=\pr(C\geq t | X=x)$.  Then the expectation of (\Rom{1}) can be written as
\$
&%E\left[\frac{1}{n}\sum_{i=1}^n \frac{I(Y_i\leq u)\delta_iK_h(B^\T X_i-z)}{f_{B^\T X(z)}E\big\{I(Y\geq Y_i)|z\big\}}  \right]
E\left[ \frac{K_b(Y_i-t)\delta_iK_h(B^\T X_i-z)}{f_{B^\T X}(z)S(t, z)E\big\{S_c(t, X_i)|z\big\}}  \right]\\
&=\iint \frac{K_b(y_i- t )K_h(z_i -z) E\{S_c(y_i, X_i|B^\T x_i)\}}{f_{B^\T X}(B^\T x) S(t, B^\T x)E \{S_c(t, X_i)|B^\T x\}}f({y_i, B^\T x_i}) f_{B^\T X}(B^\T x_i) \dif y_i \dif z_i\\
&=\iint \frac{K(v)K(u) E\{S_c(t+bv, X_i|z+hu)\}}{f_{B^\T X}(z) S(t, z)E \{S_c(t, X_i)|z\}}f({t+bv, z+hu}) f_{B^\T X}(z\!+\!hu) \dif v \dif u   \\%\tag{$z_i=z+hu$}\\
&=\iint K(v) K(u) \lambda(t, z) \dif v \dif u
\\
&\qquad + \frac{h^2\partial^2 }{2\partial  z^2}\iint \frac{K(v )K(u) E\{S_c(t, X_i|z^*)\}}{f_{B^\T X}(z) S(t, z)E \{S_c(t, X_i)|z\}}f({t, z^*}) f_{B^\T X}(z^*) u^2\dif v \dif u \\
&\qquad +\frac{b^2\partial^2 }{2\partial  t^2}\iint \frac{K(v )K(u) E\{S_c(t^*, X_i|z)\}}{f_{B^\T X}(z) S(t, z)E \{S_c(t, X_i)|z\}}f({t^*, z}) f_{B^\T X}(z)v^2 \dif v \dif u \\
%%%
%&=\lambda(t, z)+\frac{h^2\partial^2 }{2\partial  z^2}\iint \frac{I(y_i\leq t )K(u) E\{S_c(y_i, X_i|z^*)\}}{f_{B^\T X}(z) S(y_i, z)E \{S_c(y_i, X_i)|z\}}f({y_i, z^*}) f_{B^\T X}(z^*)u^2 \dif y_i \dif u\\
%&\qquad +\frac{b^2\partial^2 }{2\partial  t^2}\iint \frac{K(v )K(u) E\{S_c(t^*, X_i|z)\}}{f_{B^\T X}(z) S(t, z)E \{S_c(t, X_i)|z\}}f({t^*, z}) f_{B^\T X}(z)v^2 \dif v \dif u \\
&=\lambda(t, z)+\frac{h^2}{2} \sup\left\{   \frac{\frac{\partial^2}{\partial z^2} \left(f({t, z^*}) f_{B^\T X}(z^*) E\{S_c(y_i, X_i|z^*)\}\right)}{f_{B^\T X}(z) S(t, z)E \{S_c(t, X_i)|z\}}\right\} \int K(u)u^2  \dif u\\
&\qquad+\frac{b^2}{2} \sup\left\{   \frac{\frac{\partial^2}{\partial b^2} \left(f({t^*, z})  E\{S_c(t^*, X_i|z)\}\right)}{ S(t, z)E \{S_c(t, X_i)|z\}}\right\} \int K(v)v^2  \dif u\\
&=\lambda(t, z)+O(h^2+b^2),
\$
where we set $z_i=z+hu$ and $y_i=t+bv$ in the second equality,  $z^*$ and $t^*$ are some convex combinations of $z$ and $z+hu$, $t$ and $t+bv$ respectively. The last equality is a consequence of the assumed assumptions.

%\comment{SQ: Till here.}
Next, we bound the variance of $\widehat\lambda(t,z)$:
\$
\var\left\{\widehat \lambda(t, z)\right\}
&=\var\left\{ \frac{1}{n}\sum_{i=1}^n\frac{ K_b(Y_i- t)\delta_iK_h(B^\T X_i-z) }{f_{B^\T X}(z)E\big\{I(Y\geq t)| z \big\}}\big(1+O(A_i)\big)\right\}\\
&\leq 2\,\var\left\{\frac{1}{n}\sum_{i=1}^n\frac{ K_b(Y_i-t)\delta_iK_h(B^\T X_i-z) }{f_{B^\T X}(z)E\big\{I(Y\geq t)| z \big\}}\right\}\\
&\qquad\qquad +2\,\var\left\{\frac{1}{n}\sum_{i=1}^n\frac{ K_b(Y_i-t)\delta_iK_h(B^\T X_i-z) }{f_{B^\T X}(z)E\big\{I(Y\geq t)| z \big\}}O(A_i)\right\}.
%&\leq \frac{}{Y}
\$
We upper  bound the first quantity first.
\$
&~~2\,\var\left\{\frac{1}{n}\sum_{i=1}^n\frac{ K_b(Y_i-t)\delta_iK_h(B^\T X_i-z) }{f_{B^\T X}(z)E\big\{I(Y\geq t)| z \big\}}\right\}\\
&=\frac{2}{n}\left(E\left\{\frac{K_b(Y_i-t)\delta_i K_h(B^\T X_i-z)}{f_{B^\T X}(z)E\big\{I(Y\geq t |z)\big\}}\right\}^2-\lambda^2(t,z)\right)+O(h^2/n+b^2/n)\\
&=\frac{2}{n}E\left\{\frac{K_b(Y_i-t)\delta_i K_h(B^\T X_i-z)}{f_{B^\T X}(z)E\big\{I(Y\geq ti |z)\big\}}\right\}^2+O\left({1}/{n}\right)\\
&=\frac{1}{nhb}\iint \frac{K^2(v)K^2(u) E\{S_c(t+bv, X_i|z\!+\!hu)\}}{f^2_{B^\T X}(z) S^2(t, z)E^2 \{S_c(t, X_i)|z\}}f({t+bv, z\!+\!hu}) f_{B^\T X}(z\!+\!hu) \dif v\dif u \!+\!O(1/n)\\
&= \frac{1}{nhb}\iint \frac{K^2(v)K^2(u) }{f_{B^\T X}(z) S^2(t, z)E \{S_c(t, X_i)|z\}}f({t, z})  \dif v \dif u \\
&\qquad + \frac{h\partial^2}{nb\partial z^2}\iint \frac{K^2(v)K^2(u) E\{S_c(t, X_i|z^*)\}}{f^2_{B^\T X}(z) S^2(t, z)E^2 \{S_c(t, X_i)|z\}}f({t, z^*}) f_{B^\T X}(z^*) u^2\dif v \dif u  \\
&\qquad + \frac{b\partial^2}{nh\partial t^2} \iint \frac{K^2(v)K^2(u) E\{S_c(t^*, X_i|z)\}}{f^2_{B^\T X}(z) S^2(t, z)E^2 \{S_c(t, X_i)|z\}}f({t^*,z}) f_{B^\T X}(z) u^2\dif v \dif u + O(1/n)\\
&=O\left(1/(nbh)\right),
\$
uniformly over all $t$ and $z$.
Following the similar argument, we can show that
\$
 &2\,\var\left\{\frac{1}{n}\sum_{i=1}^n\frac{ K_b(Y_i-t)\delta_iK_h(B^\T X_i-z) }{f_{B^\T X}(z)E\big\{I(Y\geq t)| z \big\}}O(A_i)\right\}\\
& =O\left(\frac{1}{nhb}O^2(A_i)\right)=O\left((nhb)^{-2}b+(nhb)^{-1}h^4\right).
\$
In the last equality, we use the fact that $A_i^2=O_p\left(1/(nh)+h^4\right)$ uniformly, which can be proved by following a simplified argument in \cite{mack1982weak}.  Combining the above results together, we obtain
\$
|\widehat\lambda(t,z)-\lambda(t,z)|=O_p\left((nhb)^{-1/2}+h^2+b^2\right),
\$
uniformly for all $t$ and $z$. Using a similar argument, after a some long algebra,   we can prove the second and the third statements. We do not repeat the details here.

\end{proof}

\subsection{Proof of Theorem \ref{thm1}} % {\color{red} RZ: check possible typos and summarize the main result in a theoretical result section \ref{sec:theory}}
%Using the arguments of \citet[Lemma 3.5]{pakes1989simulation}, we obtain the result on the GMM estimator. See also Corollary 1 of \citet{chen2003estimation}.

%\begin{corollary} Let $B = \arg\min_{B} M_n(B, \hat{h}) \tilde{W}_n M_n(B, \hat{h})$, where $\tilde{W}_n = W_n\{\tilde{B}, \hat{h}(\tilde{B})\}$, $\tilde{B} - B_0 = {\rm o}_p(1)$, and $\{W_n(B, h)\}$ is a family of random matrices such that
%$$\sup_{\|B - B_0\| < \delta_n, \|h - h_0\| < \delta_n} \|W_n(B, h) - W\|= {\rm o}_p(1)$$
%whenever $\{\delta_n\}$ is a sequence of positive numbers that converges to zero. Under the conditions of Theorem \ref{thm1}, $$\sqrt{n} \textnormal{vec} \big(\widehat B -B\big)  \xrightarrow{d} \cN (0,\Omega),$$
%where $\Omega = (A^\T W A)^{-1} A^\T W \Sigma W A (A^\T W A)^{-1}$.
%\end{corollary}

%\begin{lemma} Let $\{G_n(\theta: \theta \in \Theta)\}$ be a family of nonsingular, random matrices for which there exists a nonsingular, nonrandom matrix $G$ such that
%$$\sup_{\|\theta - \theta_0\| < \delta_n} = {\rm o}_p(1)$$
%whenever $\{\delta_n\}$ is a sequence of positive numbers that converges to zero. Under the conditions of Theorem \ref{thm1},
%
%\end{lemma}
% We assume  the estimator $\widehat B$ is achieved by solving the following semiparametric estimating equations
%\$
%\frac{1}{n}\vecc{\bigg[\sum_{i=1}^n \int_0^\tau \Big\{ \alpha(u,X_i) - \widehat\alpha^*(u, \widehat B^\T X_i)\Big\}d\widehat M(u, \widehat B^\T X_i)\bigg]}=0,
%\$
%which covers both the inverse and forward regression schemes.
%
\setcounter{theorem}{0}
Recall that  $\beta_\ell=\vecl(B)$ and let $\beta=\vecc(B)$. Let $\tilde \beta_\ell=\vecl\big(\widetilde B\big)$ and  $\Omega=\big\{\widetilde B: \|\widetilde\beta_\ell-\beta_\ell\|_2\leq Cn^{-1/2}\big\}$, or equivalently  $\Omega=\big\{\widetilde\beta: \|\widetilde\beta_\ell-\beta_\ell\|_2\leq Cn^{-1/2}\big\}$  for some $C>0$.
%We need the following lemma, which is due to Lenglart \citep{fleming2011counting}.
%\begin{lemma}[Lenglart's Inequality]\label{lemma:1}
%Suppose $H$ is a predictable and locally bounded process. Then for any stopping time $\tau$ such that $\textnormal{pr}(\tau<\infty)=1$, and any $\epsilon, \eta>0$,
%\$
%\textnormal{pr} \bigg[\sup_{t\leq \tau}\bigg\{\int_0^t H(u)dM(u)\bigg\}^2\!\geq \epsilon\bigg] \leq \frac{\eta}{\epsilon}+\textnormal{pr}\bigg\{\int_0^\tau H^2(u)d\langle M,M\rangle(u)\geq \eta\bigg\}.
%\$
%\end{lemma}
%Let $G(u, B, X)=\vecc \big\{ \alpha(u, X)-\alpha^*(u,  B^\T X)\big\}\big\{\vecc \big(X\nabla_2\lambda(u, B^\T X)\big)\big\}^\T$ and $A(u, B, X)=\vecc\big\{ \alpha(u, X)-\alpha^*(u, B^\T X)\big\} \big[\vecc\big\{ \alpha(u, X)-\alpha^*(u, B^\T X)\big\}\big]^\T$. %Now we are ready to present the main theorem.
 %For any matrix $A$, let $\vecc(A)$ denote the vector concatenating the columns of $A$. Then we are ready to present the main theory.
We are ready to prove the main theorem of this paper.
\begin{proof}[Proof of Theorem \ref{thm1}]
Without loss of generality, we prove the theorem under the generic case in which   $\alpha^*(u, B^\T X)$ is arbitrary. For simplicity, let
\$
S_n(\widehat B)
& = \frac{1}{n}\vecc{\bigg[\sum_{i=1}^n \int_0^\tau \Big\{ \alpha(u,X_i) - \widehat\alpha^*(u, \widehat B^\T X_i)\Big\}d\widehat M(u, \widehat B^\T X_i)\bigg]}.
\$  %The proof for the inverse regression case follows from the similar spirit, and thus  is omitted here. %{\color{red} Q: Not necessarily the same expression for asymptotic normality as here.}
We sometimes write $M(u, B^\T X_i)$ as $M_i(u)$ and  $S_n(B)$ as $S_n(\beta_\ell)$. %By the asymptotic theory of the $M$-estimation and
Following \cite{jureckova1971nonparametric} and \cite{tsiatis1990estimating}, it suffices to show that $S_n(\beta_\ell)$ is uniformly asymptotically linear in a small neighborhood of the true parameter $\beta_\ell$. That is, we need to show, there exists some linear operator, $G_n$, such that
\begin{gather*}
\sup_{\|\widehat \beta_\ell-\beta_\ell\|_2\leq Cn^{-1/2}} n^{1/2}\big\|S_n(\widehat\beta_\ell)-S_n^*(\beta_\ell)-G_n(\widehat\beta_\ell-\beta_\ell)\big\|_2= {\rm o}_p\big(1\big),~\textnormal{where}\\
S_n^*(\beta_\ell)=S_n^*(B) = \frac{1}{n}\vecc{\bigg[\sum_{i=1}^n \int_0^\tau \Big\{ \alpha(u,X_i) -\alpha^*(u,  B^\T X_i)\Big\}dM(u,  B^\T X_i)\bigg]}.
\end{gather*}

 We then expand  $S_n(\widehat \beta_\ell)$ as follows,
\#\label{Q:A2}
\!\!0&= S_n(\widehat \beta_\ell) = \frac{1}{n}\vecc{\bigg[\sum_{i=1}^n \int_0^\tau \Big\{ \alpha(u,X_i) - \widehat\alpha^*(u, \widehat B^\T X_i)\Big\}d\widehat M(u, \widehat B^\T X_i)\bigg]}\notag \\
&=\underbrace{ \frac{1}{n}\vecc{\bigg[\sum_{i=1}^n \int_0^\tau \Big\{ \alpha(u,X_i) - \widehat \alpha^*(u, B^\T X_i)\Big\}d \widehat M(u,  B^\T X_i)\bigg]}}_{\Rom{1}}  \notag\\
 &~+\underbrace{\frac{1}{n}\sum_{i=1}^n\frac{\partial }{\partial z}\vecc{\left[ \big\{\alpha(u, X_i)-\widehat \alpha^*(u,  B^\T X_i) \big\} d \widehat M(u,  B^\T X_i)   \right]} (\widehat B-B)^\T X_i}_{\Rom{2}} +o_p\left(n^{-\frac{1}{2}}\right),\notag
\#
%where $\widetilde B$ is some convex combination of $\widehat B$ and $B$.
where $z$ is a vectorized argument of $B^\T X$. It remains to analyze terms (\Rom{1}) and (\Rom{2}) respectively. For term (\Rom{1}), we obtain that
\$
\Rom{1}& = {\frac{1}{n}\vecc \bigg[\sum_{i=1}^n \int_0^\tau\Big\{ \alpha(u, X_i)-\alpha^*(u,  B^\T X_i)\Big\}dM(u, B^\T X_i)\Big\} \bigg]}     \notag\\
&~+\underbrace{\frac{1}{n}\vecc{\bigg[\sum_{i=1}^n \int_0^\tau \Big\{ \alpha(u,X_i)-\alpha^*(u, B^\T X_i)\Big\}d\Big(\widehat M(u, B^\T X_i)-M(u, B^\T X_i)\Big)\bigg]}}_{\Rom{1}_1}  \notag\\
&~+\underbrace{\frac{1}{n}\vecc{\bigg[\sum_{i=1}^n \int_0^\tau \Big\{ \alpha^*(u,B^\T X_i)-\widehat\alpha^*(u,  B^\T X_i)\Big\}d\Big\{\widehat M(u,  B^\T X_i)-M(u, B^\T X_i)\Big\}\bigg]}}_{\Rom{1}_2}\\
&~+\underbrace{\frac{1}{n}\vecc{\bigg[\sum_{i=1}^n \int_0^\tau \Big\{ \alpha^*(u,B^\T X_i)-\widehat\alpha^*(u,  B^\T X_i)\Big\}dM(u, B^\T X_i)\bigg]}}_{\Rom{1}_3}.
\$
%where
%$
%\widehat\alpha^* (u, \widehat B^\T X)={\widehat{E} \big\{ \alpha(u,X)Y(u)|\widehat B^\T X\big\}}\big/\widehat E\big\{Y(u)|\widehat B^\T X\big\}.
%$
We bound terms $\Rom{1}_1$, $\Rom{1}_2$ and $\Rom{1}_3$, respectively. For  $\Rom{1}_1$, we can applied Lemma \ref{lem:thm.1}. Note that for certain choice of the bandwidth $b$ and $h$, we can always find some $\kappa < 1/2$ such that $\widehat\lambda(t| z)=\lambda(t, z)+O_p(n^{-1/2 + \kappa})$. Hence, for simplification, we use $\kappa$ throughout the rest of this proof.
\$
\Rom{1}_1&=\frac{1}{n}\vecc{\bigg[\sum_{i=1}^n \int_0^\tau \Big\{ \alpha(u,X_i)-\alpha^*(u, B^\T X_i)\Big\} \left\{\lambda (u, B^\T X_i)-\widehat \lambda(u, B^\T X_i)\right\}Y_i(u)du}\\
&=O_p\left(n^{-1+\kappa} \right)O_p\bigg(n^{-\frac{1}{2}} \vecc{\bigg[\sum_{i=1}^n \int_0^\tau \Big\{ \alpha(u,X_i)-\alpha^*(u, B^\T X_i)\Big\} \lambda (u, B^\T X_i)Y_i(u)du} \bigg] \bigg)\\
&=O_p(n^{-1+\kappa})O_p\left(1\right)=O_p(n^{-1+\kappa}),
\$
where the last inequality is due to central limit theorem. Following the some algebra and using Assumption \ref{ass:4}, we can show that $\Rom{1}_2=O_p(n^{-1+\kappa})$ and $\Rom{1}_3=O_p(n^{-1+\kappa})$. Combining the bounds for the above three terms above, we obtain that
\$
\Rom{1}& = {\frac{1}{n}\vecc \bigg[\sum_{i=1}^n \int_0^\tau\big\{ \alpha(u, X_i)-\alpha^*(u,  B^\T X_i)\big\}dM(u, B^\T X_i) \bigg]}+O_p\left(n^{-1+\kappa}\right)\\
&= {\frac{1}{n}\vecc \bigg[\sum_{i=1}^n \int_0^\tau\big\{ \alpha(u, X_i)-\alpha^*(u,  B^\T X_i)\big\}dM(u, B^\T X_i) \bigg]}+o_p\left(n^{-1/2}\right),
\$
where we use the fact that $\kappa<1/2$ with proper choice of $b$ and $h$. For term $\Rom{2}$, in a similar argument, we shall  obtain  that
\$
\Rom{2}& = \frac{1}{n}\sum_{i=1}^n\frac{\partial }{\partial z}{\left(\int_0^\tau \vecc{\left[ \big\{\alpha(u, X_i)- \alpha^*(u,  B^\T X_i) \big\}\right]} d  M(u,  B^\T X_i)   \right)} (\widehat B-B)^\T X_i+o_p\left(n^{-1/2}\right)\\
&=\frac{1}{n}\sum_{i=1}^n\frac{\partial }{\partial \beta_\ell}{\left(\int_0^\tau \vecc{\left[ \big\{\alpha(u, X_i)- \alpha^*(u,  B^\T X_i) \big\}\right]} d  M(u,  B^\T X_i)   \right)} (\widehat \beta_\ell-\beta_\ell)+o_p\left(n^{-1/2}\right),
\$
where we rewrite the first term in the right hand side as the derivative with respect to $\beta_\ell$. For simplicity, define
\$
A_i(\tau)&=\int_0^\tau\vecc\left\{\alpha(u, X_i)-\alpha^*(u,  B^\T X_i)\right\}dM(u, B^\T X_i),\\
\textnormal{and} \quad G_n&=\frac{1}{n}\sum_{i=1}^n\frac{\partial }{\partial \beta_\ell}{\left(\int_0^\tau \vecc{\left[ \big\{\alpha(u, X_i)- \alpha^*(u,  B^\T X_i) \big\}\right]} d  M(u,  B^\T X_i)   \right)}.
\$
Combing the asymptotically linear expansion with terms $\Rom{1}$ and $\Rom{2}$, $S_n(\widehat B)$ can be written as
\$
\widehat S_n(\widehat B)
%&=\underbrace{n^{-1}\vecc\bigg[ \sum_{i=1}^n\int_0^\tau\Big\{ \alpha(u,X_i)-\alpha^*(u, \widehat B^\T X_i)\Big\}dM_i(u) \bigg]}_{S^*_n(\widehat B)}+{\rm o}_p\bigg(\sqrt{\frac{1}{n}}\bigg)\notag\\
&=\frac{1}{n} \sum_{i=1}^nA_i(\tau)+G_n\Big (\beta_\ell-\widehat\beta_\ell\Big)   +{\rm o}_p\big(n^{-1/2}\big)=0.
%&=\underbrace{\frac{1}{n}\vecc\bigg[ \sum_{i=1}^n\int_0^\tau\Big\{ \alpha(u,X_i)-\alpha^*(u,  B^\T X_i)\Big\}dM_i(u) \bigg]}_{S^*_n(B)}+  G\big (\beta-\widehat\beta\big)   +{\rm o}_p\bigg(\sqrt{\frac{1}{n}}\bigg),
\$
 Therefore, $\widehat\beta_\ell-\beta_\ell$ can be written as
\$
\sqrt{n}\big(\widehat\beta_\ell-\beta_\ell\big)\!=\!(G_n^\T G_n)^{-1}G_n^\T\bigg\{\frac{1}{\sqrt n} \sum_{i=1}^nA_i(\tau)\bigg\}+{\rm o}_p(1),
\$
where we implicitly assume that  $G_n^\T G_n$ is invertible, which requires that at the number of estimating equations is larger than the number of parameters. In other words, this assume that the estimating equations are rich enough to recover the $d(p\!-\!d)$-dimensional vector of parameters $\beta_\ell$. This finishes the proof.
\end{proof}

\section{Computational Time}

This computational time for all proposed method is presented in Table \ref{tab:time}. All simulations are done on an Intel Xeon E5-2680v4 processor with 5 parallel threads.

\begin{table}[htp]
\def~{\hphantom{0}}
\centering
\footnotesize
\begin{threeparttable}
\caption{Mean computational time (in seconds)}
\begin{tabular}{l rrrrrr}
& \multicolumn{3}{c}{Setting 1} & \multicolumn{3}{c}{Setting 2} \\
Dimension & \multicolumn{1}{c}{$p=6$} & \multicolumn{1}{c}{$p=12$} & \multicolumn{1}{c}{$p=18$} & \multicolumn{1}{c}{$p=6$} & \multicolumn{1}{c}{$p=12$} & \multicolumn{1}{c}{$p=18$} \\
Forward	&	$< 1$	&	1	&	1	&	$< 1$	&	1	&	1	\\
CP-SIR	&	$< 1$	&	$< 1$	&	$< 1$	&	$< 1$	&	$< 1$	&	$< 1$	\\
IR-CP	&	$< 1$	&	1	&	1	&	2	&	6	&	11	\\
IR-Semi	&	1	&	4	&	13	&	8	&	56	&	162	\\
{\smallskip\vspace{-0.1in}}\\
& \multicolumn{3}{c}{Setting 3} & \multicolumn{3}{c}{Setting 4} \\
Dimension & \multicolumn{1}{c}{$p=6$} & \multicolumn{1}{c}{$p=12$} & \multicolumn{1}{c}{$p=18$} & \multicolumn{1}{c}{$p=6$} & \multicolumn{1}{c}{$p=12$} & \multicolumn{1}{c}{$p=18$} \\
Forward	&	$< 1$	&	$< 1$	&	1	&	$< 1$	&	$< 1$	&	1	\\
CP-SIR	&	$< 1$	&	$< 1$	&	$< 1$	&	$< 1$	&	$< 1$	&	$< 1$	\\
IR-CP	&	2	&	5	&	8	&	1	&	3	&	7	\\
IR-Semi	&	12	&	47	&	146	&	6	&	27	&	96	\\
\end{tabular}\label{tab:time}
\scriptsize
\begin{tablenotes}
\item Forward: forward regression; CP-SIR: the computational efficient approach; IR-CP: the counting process inverse regression approach; IR-Semi: the semiparametric inverse regression approach. All simulations are done on an Intel Xeon E5-2680v4 processor with 5 parallel threads.
\end{tablenotes}
\end{threeparttable}
\end{table}

\section{Additional Simulation Results} \label{sec:add.sim}

This section includes parameter estimations and standard deviation results in the simulation. For each estimated parameter matrix $\widehat B$, we linearly transform the columns of this matrix such that the block sub-matrix is an identity matrix. For example, in settings 2 and 3, we require the first two rows to be a block diagonal matrix, i.e., the transformed parameter matrix is obtained through $\widehat B \widehat B_{[1:2,\, 1:2]}^{-1}$, where $B_{[1:2,\, 1:2]}$ is the upper-block (first two rows) of $\widehat B$. In setting 1, a one dimensional case, only the first row is used, hence is essentially $\widehat B \widehat B_{[1,\, 1]}^{-1}$. In setting 4, the first and the third rows are used. Parameter estimations of other rows are reported in the following tables. To estimate the standard deviation, we use 100 bootstrapped samples of the training data to estimate the parameters, then translate each parameter estimate into the form with a diagonal matrix at the specified rows. Then for each parameter in the rest of matrix, since the parameters asymptotically follow a joint Gaussian distribution, we use 1.4826 times the median absolute deviation of the 100 bootstrap replicates as an estimation of the standard error, following the fact that this is a consistent estimator of the standard deviation in the normal distribution. The reason that we use the median absolute deviation instead of a regular standard deviation estimator is to ensure robustness such that extreme values (possibly due to convergence issue caused by duplicated samples) in the bootstrap estimations do not dominate the result.

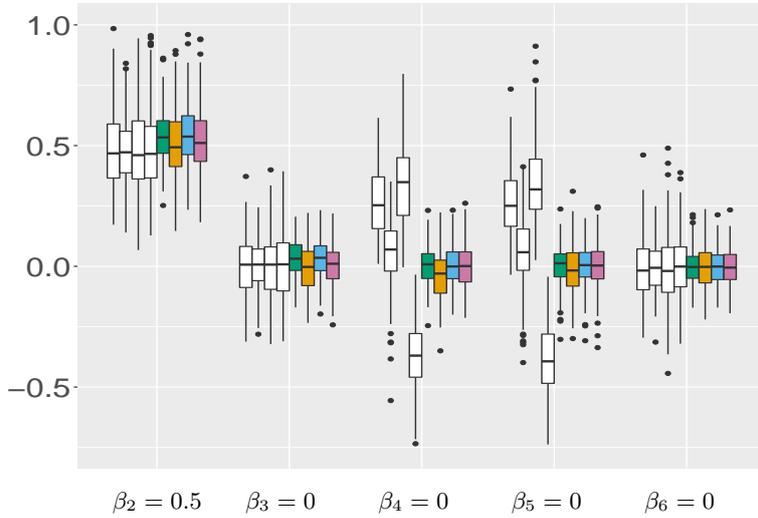
\begin{figure}[htp]
\centering
\footnotesize
\caption{Boxplot of parameter estimates (Setting 1, $p = 6$)}
    \begin{tikzpicture}
    \node[inner sep=0pt] (dm) at (0,0)
        {\includegraphics[width=4in, height = 2.5in]{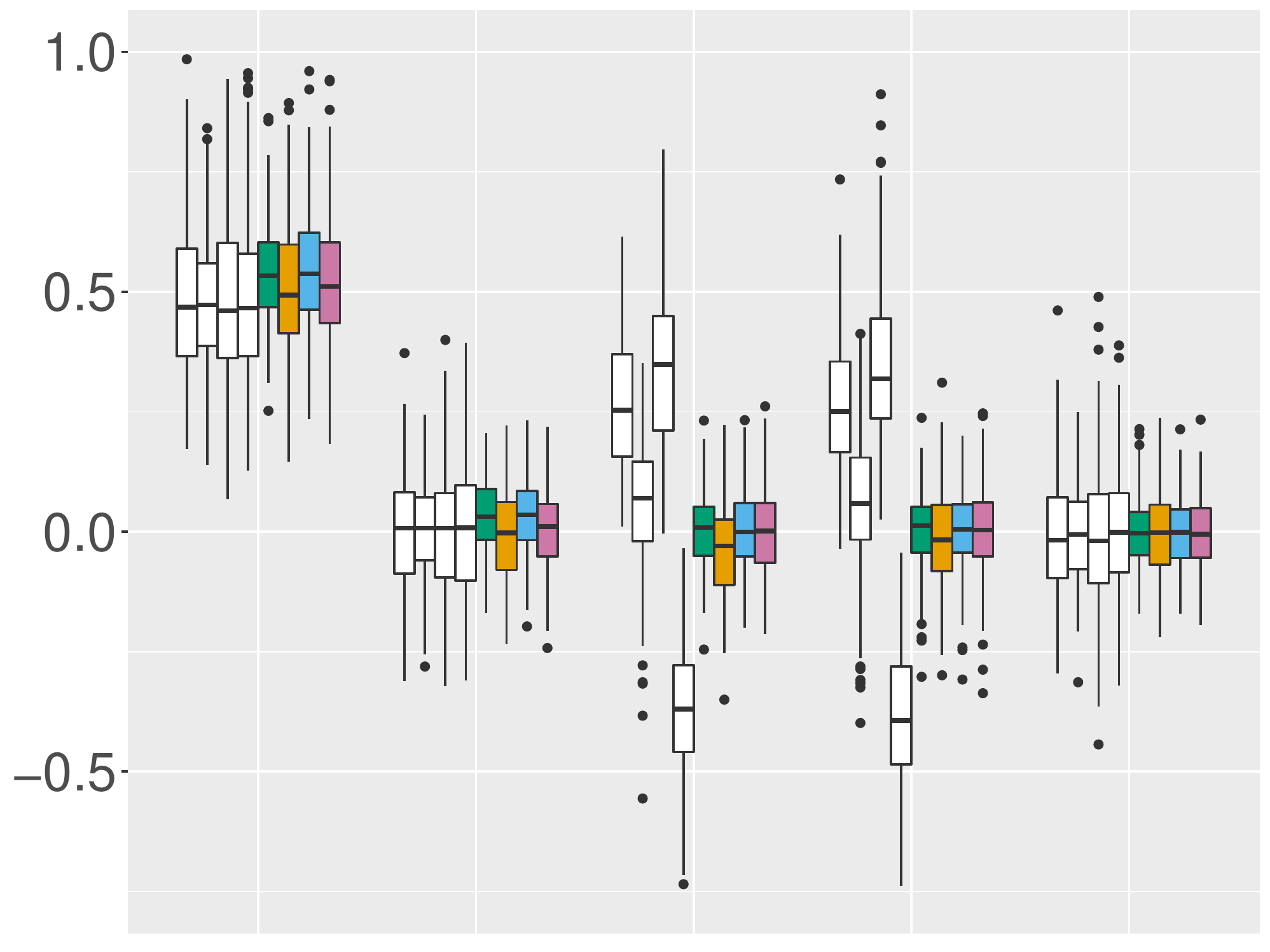}};
    \node[text width=6cm] at (-0.6, -3.5)
        {$\beta_2 = 0.5$};
    \node[text width=6cm] at (1.15, -3.5)
        {$\beta_3 = 0$};
    \node[text width=6cm] at (2.9, -3.5)
        {$\beta_4 = 0$};
    \node[text width=6cm] at (4.65, -3.5)
        {$\beta_5 = 0$};
    \node[text width=6cm] at (6.4, -3.5)
        {$\beta_6 = 0$};
    \end{tikzpicture}\label{fig:sds1}
    \scriptsize For each paramter, the boxes are (from left to right, with proposed methods colored): DS \citep{li1999dimension}; IPCW-SIR \citep{lu2011sufficient}; hMave \citep{xia2010dimension}; The forward regression; The computational efficient approach; The counting process inverse regression; The semiparametric inverse regression.
\end{figure}

\begin{table}[htp]
\def~{\hphantom{0}}
\renewcommand{\arraystretch}{0.75}
\centering
\begin{threeparttable}
\scriptsize
\caption{Mean ($\times 10^4$) and standard deviation ($\times 10^4$) of parameter estimations (Setting 1, $p=6$).}
\begin{tabular}{r p{0.8cm} rrrrr}
&  & $\beta_2 = 0.5$ & $\beta_3 = 0$ & $\beta_4 = 0$ & $\beta_5 = 0$ & $\beta_6 = 0$ \\
{\smallskip\vspace{-0.1in}}\\
Naive	&	mean	&	480	&	-3	&	267	&	263	&	-13	\\
	&	sd	&	158	&	127	&	137	&	140	&	120	\\
DS	&	mean	&	477	&	2	&	53	&	55	&	-6	\\
	&	sd	&	127	&	103	&	138	&	146	&	93	\\
IPCW-SIR	&	mean	&	482	&	-3	&	343	&	352	&	-17	\\
	&	sd	&	185	&	136	&	166	&	173	&	146	\\
hMave	&	mean	&	482	&	2	&	-380	&	-394	&	0	\\
	&	sd	&	178	&	145	&	155	&	161	&	123	\\
Forward	&	mean	&	536	&	36	&	-3	&	-1	&	3	\\
	&	sd	&	98	&	80	&	81	&	81	&	75	\\
	&	$\widehat{\text{sd}}$	&	99	&	80	&	82	&	82	&	74	\\
	&	coverage	&	956	&	920	&	940	&	948	&	952	\\
CP-SIR	&	mean	&	511	&	0	&	-37	&	-29	&	2	\\
	&	sd	&	129	&	104	&	102	&	94	&	90	\\
	&	$\widehat{\text{sd}}$	&	134	&	105	&	103	&	105	&	94	\\
	&	coverage	&	960	&	964	&	956	&	960	&	972	\\
IR-CP	&	mean	&	546	&	45	&	-3	&	-4	&	1	\\
	&	sd	&	108	&	86	&	88	&	93	&	79	\\
	&	$\widehat{\text{sd}}$	&	108	&	88	&	91	&	89	&	82	\\
	&	coverage	&	934	&	924	&	932	&	952	&	968	\\
IR-Semi	&	mean	&	520	&	13	&	-5	&	-4	&	1	\\
	&	sd	&	115	&	91	&	95	&	99	&	80	\\
	&	$\widehat{\text{sd}}$	&	117	&	94	&	96	&	98	&	85	\\
	&	coverage	&	932	&	960	&	948	&	944	&	956	\\
\end{tabular}\label{tab:sds1}
\begin{tablenotes}
\item \scriptsize DS: \cite{li1999dimension}; IPCW-SIR: \cite{lu2011sufficient}; hMave: \cite{xia2010dimension}; Forward: forward regression; CP-SIR: the computational efficient approach; IR-CP: the counting process inverse regression; IR-Semi: the semiparametric inverse regression. For the proposed methods, ``$\widehat{\text{sd}}$'' is the bootstrap estimation of the standard deviation, and ``coverage'' is the coverage rate of the 95\% confidence internal.
\end{tablenotes}
\end{threeparttable}
\end{table}

\begin{figure}[htp]
\centering
\footnotesize
\caption{Boxplot of parameter estimates (Setting 2, $p = 6$)}
    \begin{tikzpicture}
    \node[inner sep=0pt] (dm) at (0,0)
        {\includegraphics[width=6in, height = 2.5in]{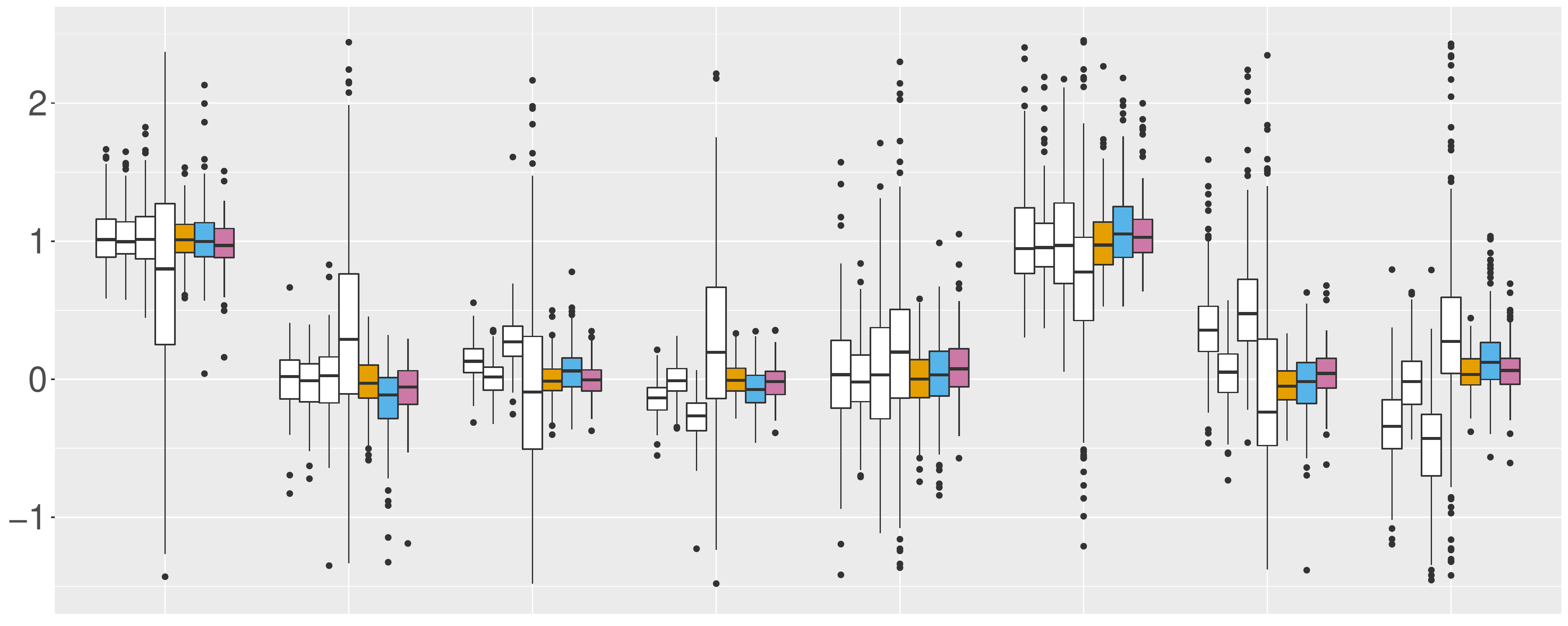}};
    \node[text width=6cm] at (-3.75, -3.5)
        {$\beta_{13} = 1$};
    \node[text width=6cm] at (-1.95, -3.5)
        {$\beta_{14} = 0$};
    \node[text width=6cm] at (-0.15, -3.5)
        {$\beta_{15} = 0$};
    \node[text width=6cm] at (1.65, -3.5)
        {$\beta_{16} = 0$};
    \node[text width=6cm] at (3.45, -3.5)
        {$\beta_{23} = 0$};
    \node[text width=6cm] at (5.25, -3.5)
        {$\beta_{24} = 1$};
    \node[text width=6cm] at (7.05, -3.5)
        {$\beta_{25} = 0$};
    \node[text width=6cm] at (8.85, -3.5)
        {$\beta_{26} = 0$};
    \end{tikzpicture}\label{fig:sds2}
    \scriptsize For each paramter, the boxes are (from left to right, with proposed methods colored): DS \citep{li1999dimension}; IPCW-SIR \citep{lu2011sufficient}; hMave \citep{xia2010dimension}; The forward regression; The computational efficient approach; The counting process inverse regression; The semiparametric inverse regression. Proposed methods are colored.
\end{figure}

\begin{table}[htbp]
\def~{\hphantom{0}}
\renewcommand{\arraystretch}{0.75}
\centering
\begin{threeparttable}
\scriptsize
\caption{Mean ($\times 10^4$) and standard deviation ($\times 10^4$) of parameter estimations (Setting 2, $p=6$).}
\begin{tabular}{r p{0.8cm} rrrrrrrr}
&  & $\beta_{13} = 1$ & $\beta_{14} = 0$ & $\beta_{15} = 0$ & $\beta_{16} = 0$
& $\beta_{23} = 0$ & $\beta_{24} = 1$ & $\beta_{25} = 0$ & $\beta_{26} = 0$ \\
{\smallskip\vspace{-0.1in}}\\
Naive	&	mean	&	1029	&	-3	&	134	&	-138	&	37	&	1101	&	377	&	-351	\\
	&	sd	&	206	&	209	&	135	&	126	&	412	&	1080	&	317	&	302	\\
DS	&	mean	&	1027	&	-31	&	12	&	-8	&	-2	&	996	&	38	&	-21	\\
	&	sd	&	196	&	195	&	130	&	121	&	265	&	289	&	217	&	223	\\
IPCW-SIR	&	mean	&	1033	&	-14	&	275	&	-272	&	89	&	1347	&	551	&	-504	\\
	&	sd	&	230	&	301	&	185	&	159	&	785	&	3514	&	1123	&	660	\\
hMave	&	mean	&	1294	&	413	&	1742	&	-2036	&	-634	&	1245	&	-2692	&	2902	\\
	&	sd	&	12873	&	11019	&	35093	&	36345	&	10599	&	8001	&	31410	&	32855	\\
CP-SIR	&	mean	&	989	&	3	&	-17	&	11	&	9	&	979	&	-8	&	23	\\
	&	sd	&	194	&	175	&	134	&	112	&	231	&	220	&	169	&	148	\\
	&	$\widehat{\text{sd}}$	&	187	&	194	&	140	&	127	&	248	&	257	&	182	&	161	\\
	&	coverage	&	926	&	959	&	950	&	988	&	983	&	975	&	979	&	983	\\
IR-CP	&	mean	&	980	&	-125	&	31	&	-38	&	45	&	1127	&	7	&	106	\\
	&	sd	&	266	&	347	&	193	&	168	&	466	&	556	&	335	&	262	\\
	&	$\widehat{\text{sd}}$	&	270	&	321	&	208	&	188	&	438	&	507	&	338	&	293	\\
	&	coverage	&	967	&	1000	&	983	&	1000	&	992	&	992	&	979	&	992	\\
IR-Semi	&	mean	&	960	&	-37	&	-21	&	-10	&	79	&	1054	&	64	&	46	\\
	&	sd	&	189	&	176	&	123	&	117	&	224	&	269	&	194	&	154	\\
	&	$\widehat{\text{sd}}$	&	196	&	215	&	148	&	133	&	289	&	317	&	228	&	201	\\
	&	coverage	&	963	&	983	&	979	&	983	&	992	&	988	&	975	&	983	\\
\end{tabular}\label{tab:sds2}
\begin{tablenotes}
\item \scriptsize DS: \cite{li1999dimension}; IPCW-SIR: \cite{lu2011sufficient}; hMave: \cite{xia2010dimension}; Forward: forward regression; CP-SIR: the computational efficient approach; IR-CP: the counting process inverse regression; IR-Semi: the semiparametric inverse regression. For the proposed methods, ``$\widehat{\text{sd}}$'' is the bootstrap estimation of the standard deviation, and ``coverage'' is the coverage rate of the 95\% confidence internal.
\end{tablenotes}
\end{threeparttable}
\end{table}

\begin{figure}[htp]
\centering
\footnotesize
\caption{Boxplot of parameter estimates (Setting 3, $p = 6$)}
    \begin{tikzpicture}
    \node[inner sep=0pt] (dm) at (0,0)
        {\includegraphics[width=6in, height = 2.5in]{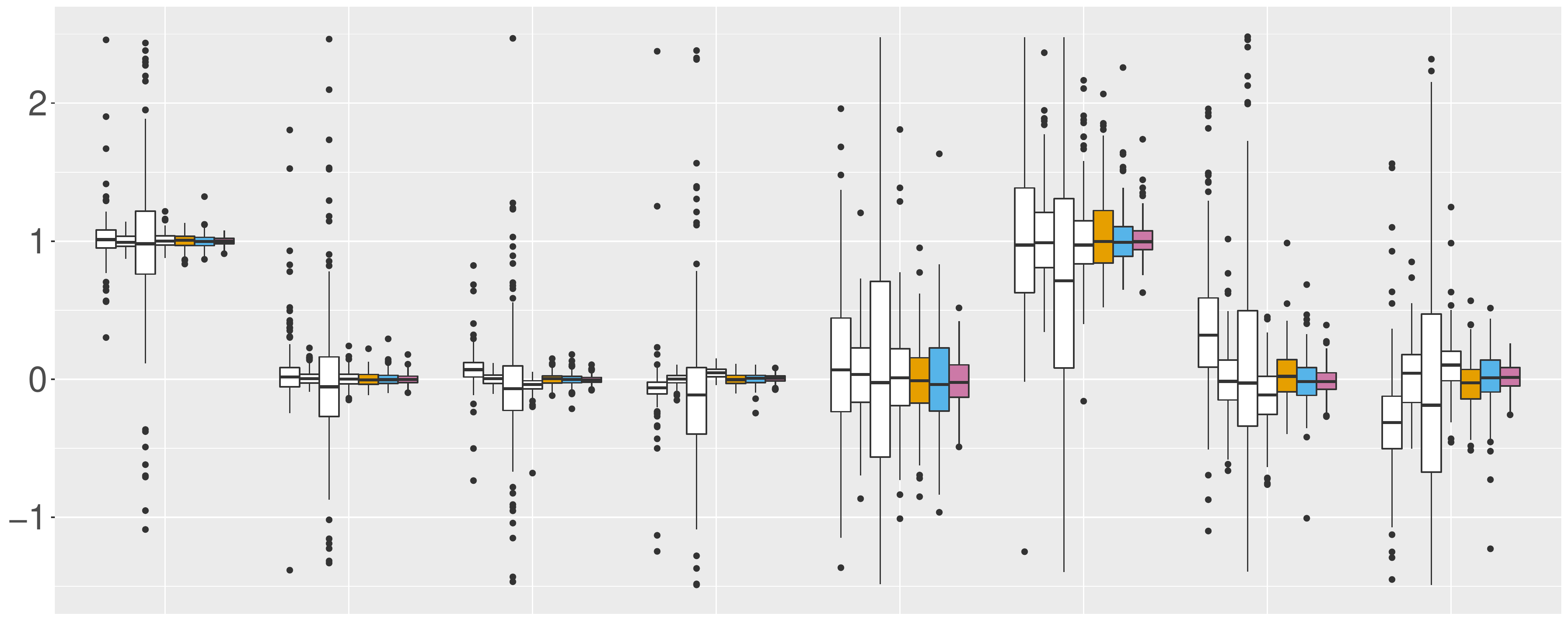}};
    \node[text width=6cm] at (-3.75, -3.5)
        {$\beta_{13} = 1$};
    \node[text width=6cm] at (-1.95, -3.5)
        {$\beta_{14} = 0$};
    \node[text width=6cm] at (-0.15, -3.5)
        {$\beta_{15} = 0$};
    \node[text width=6cm] at (1.65, -3.5)
        {$\beta_{16} = 0$};
    \node[text width=6cm] at (3.45, -3.5)
        {$\beta_{23} = 0$};
    \node[text width=6cm] at (5.25, -3.5)
        {$\beta_{24} = 1$};
    \node[text width=6cm] at (7.05, -3.5)
        {$\beta_{25} = 0$};
    \node[text width=6cm] at (8.85, -3.5)
        {$\beta_{26} = 0$};
    \end{tikzpicture}\label{fig:sds3}
    \scriptsize For each paramter, the boxes are (from left to right, with proposed methods colored): DS \citep{li1999dimension}; IPCW-SIR \citep{lu2011sufficient}; hMave \citep{xia2010dimension}; The forward regression; The computational efficient approach; The counting process inverse regression; The semiparametric inverse regression. Proposed methods are colored.
\end{figure}

\begin{table}[htbp]
\def~{\hphantom{0}}
\renewcommand{\arraystretch}{0.75}
\centering
\begin{threeparttable}
\scriptsize
\caption{Mean ($\times 10^4$) and standard deviation ($\times 10^4$) of parameter estimations (Setting 3, $p=6$).}
\begin{tabular}{r p{0.8cm} rrrrrrrr}
&  & $\beta_{13} = 1$ & $\beta_{14} = 0$ & $\beta_{15} = 0$ & $\beta_{16} = 0$
& $\beta_{23} = 0$ & $\beta_{24} = 1$ & $\beta_{25} = 0$ & $\beta_{26} = 0$ \\
{\smallskip\vspace{-0.1in}}\\
Naive	&	mean	&	1017	&	-26	&	59	&	-60	&	97	&	818	&	293	&	-271	\\
	&	sd	&	184	&	924	&	275	&	242	&	1107	&	6038	&	1787	&	1601	\\
DS	&	mean	&	1000	&	8	&	1	&	3	&	28	&	1035	&	1	&	26	\\
	&	sd	&	53	&	52	&	44	&	42	&	287	&	344	&	261	&	236	\\
IPCW-SIR	&	mean	&	1133	&	156	&	53	&	-460	&	542	&	1552	&	374	&	-1592	\\
	&	sd	&	1080	&	2164	&	1951	&	2036	&	4062	&	7741	&	6845	&	7293	\\
hMave	&	mean	&	984	&	-5	&	-46	&	33	&	-81	&	969	&	-139	&	38	\\
	&	sd	&	310	&	136	&	63	&	175	&	1890	&	848	&	325	&	1057	\\
CP-SIR	&	mean	&	998	&	8	&	-2	&	0	&	13	&	1043	&	35	&	-12	\\
	&	sd	&	56	&	53	&	42	&	43	&	278	&	313	&	189	&	226	\\
	&	$\widehat{\text{sd}}$	&	60	&	61	&	45	&	45	&	295	&	313	&	211	&	208	\\
	&	coverage	&	960	&	972	&	964	&	952	&	964	&	960	&	972	&	932	\\
IR-CP	&	mean	&	1003	&	5	&	-1	&	2	&	21	&	1012	&	-6	&	14	\\
	&	sd	&	47	&	42	&	40	&	40	&	340	&	189	&	206	&	195	\\
	&	$\widehat{\text{sd}}$	&	53	&	55	&	47	&	47	&	358	&	237	&	216	&	210	\\
	&	coverage	&	996	&	988	&	976	&	984	&	968	&	984	&	984	&	972	\\
IR-Semi	&	mean	&	1002	&	4	&	-2	&	1	&	8	&	1006	&	-6	&	7	\\
	&	sd	&	31	&	31	&	27	&	27	&	179	&	106	&	98	&	105	\\
	&	$\widehat{\text{sd}}$	&	37	&	39	&	33	&	33	&	214	&	149	&	124	&	125	\\
	&	coverage	&	984	&	992	&	988	&	984	&	984	&	988	&	996	&	976	\\
\end{tabular}\label{tab:sds3}
\begin{tablenotes}
\item \scriptsize DS: \cite{li1999dimension}; IPCW-SIR: \cite{lu2011sufficient}; hMave: \cite{xia2010dimension}; Forward: forward regression; CP-SIR: the computational efficient approach; IR-CP: the counting process inverse regression; IR-Semi: the semiparametric inverse regression. For the proposed methods, ``$\widehat{\text{sd}}$'' is the bootstrap estimation of the standard deviation, and ``coverage'' is the coverage rate of the 95\% confidence internal.
\end{tablenotes}
\end{threeparttable}
\end{table}

\begin{figure}[htp]
\centering
\footnotesize
\caption{Boxplot of parameter estimates (Setting 4, $p = 6$)}
    \begin{tikzpicture}
    \node[inner sep=0pt] (dm) at (0,0)
        {\includegraphics[width=6in, height = 2.5in]{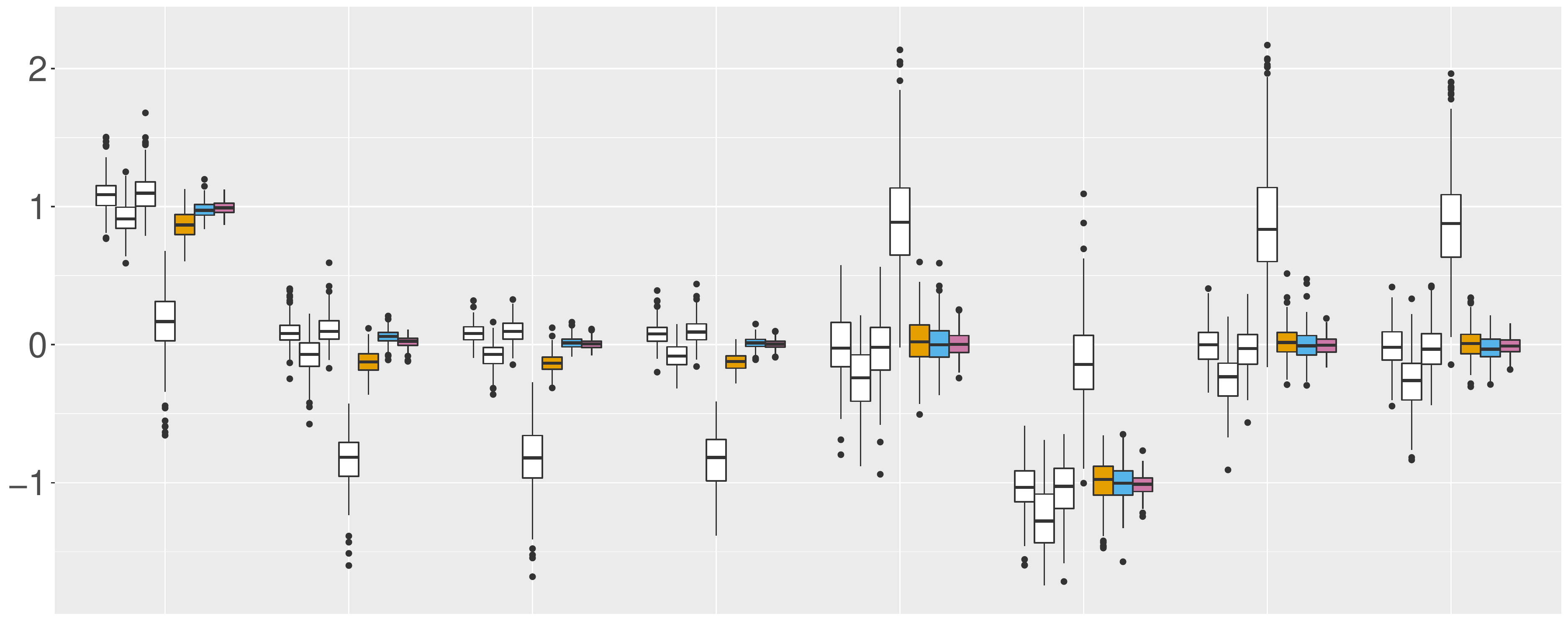}};
    \node[text width=6cm] at (-3.75, -3.5)
        {$\beta_{12} = 0.5$};
    \node[text width=6cm] at (-1.95, -3.5)
        {$\beta_{14} = 0$};
    \node[text width=6cm] at (-0.15, -3.5)
        {$\beta_{15} = 0$};
    \node[text width=6cm] at (1.65, -3.5)
        {$\beta_{16} = 0$};
    \node[text width=6cm] at (3.45, -3.5)
        {$\beta_{22} = 0$};
    \node[text width=6cm] at (5.25, -3.5)
        {$\beta_{24} = -1$};
    \node[text width=6cm] at (7.05, -3.5)
        {$\beta_{25} = 0$};
    \node[text width=6cm] at (8.85, -3.5)
        {$\beta_{26} = 0$};
    \end{tikzpicture}\label{fig:sds4}
    \scriptsize The true parameter is For each paramter, the boxes are (from left to right, with proposed methods colored): DS \citep{li1999dimension}; IPCW-SIR \citep{lu2011sufficient}; hMave \citep{xia2010dimension}; The forward regression; The computational efficient approach; The counting process inverse regression; The semiparametric inverse regression. Proposed methods are colored.
\end{figure}

\begin{table}[htbp]
\def~{\hphantom{0}}
\renewcommand{\arraystretch}{0.75}
\centering
\begin{threeparttable}
\scriptsize
\caption{Mean ($\times 10^4$) and standard deviation ($\times 10^4$) of parameter estimations (Setting 4, $p=6$).}
\begin{tabular}{r p{0.8cm} rrrrrrrr}
&  & $\beta_{12} = 1$ & $\beta_{14} = 0$ & $\beta_{15} = 0$ & $\beta_{16} = 0$
& $\beta_{22} = 0$ & $\beta_{24} = -1$ & $\beta_{25} = 0$ & $\beta_{26} = 0$ \\
{\smallskip\vspace{-0.1in}}\\
Naive	&	mean	&	1083	&	93	&	83	&	79	&	-17	&	-1043	&	-2	&	-11	\\
	&	sd	&	132	&	101	&	71	&	82	&	232	&	199	&	141	&	159	\\
DS	&	mean	&	922	&	-75	&	-79	&	-82	&	-250	&	-1293	&	-253	&	-268	\\
	&	sd	&	116	&	131	&	92	&	91	&	226	&	263	&	167	&	199	\\
IPCW-SIR	&	mean	&	1102	&	107	&	98	&	96	&	-21	&	-1062	&	-32	&	-24	\\
	&	sd	&	135	&	106	&	82	&	89	&	242	&	225	&	154	&	160	\\
hMave	&	mean	&	149	&	-837	&	-864	&	-860	&	951	&	-122	&	950	&	928	\\
	&	sd	&	239	&	194	&	310	&	265	&	419	&	329	&	539	&	443	\\
CP-SIR	&	mean	&	873	&	-127	&	-122	&	-125	&	34	&	-989	&	19	&	13	\\
	&	sd	&	102	&	91	&	63	&	70	&	199	&	164	&	126	&	140	\\
	&	$\widehat{\text{sd}}$	&	99	&	92	&	71	&	70	&	192	&	178	&	137	&	135	\\
	&	coverage	&	672	&	692	&	620	&	560	&	932	&	948	&	956	&	940	\\
IR-CP	&	mean	&	975	&	53	&	18	&	16	&	8	&	-1007	&	-2	&	-3	\\
	&	sd	&	62	&	57	&	40	&	42	&	168	&	144	&	104	&	118	\\
	&	$\widehat{\text{sd}}$	&	71	&	65	&	49	&	48	&	212	&	175	&	138	&	139	\\
	&	coverage	&	944	&	876	&	952	&	960	&	992	&	984	&	992	&	972	\\
IR-Semi	&	mean	&	990	&	19	&	9	&	5	&	10	&	-1007	&	0	&	-6	\\
	&	sd	&	50	&	45	&	33	&	34	&	91	&	82	&	64	&	68	\\
	&	$\widehat{\text{sd}}$	&	54	&	49	&	37	&	37	&	120	&	106	&	81	&	82	\\
	&	coverage	&	952	&	952	&	964	&	948	&	992	&	972	&	976	&	972	\\
\end{tabular}\label{tab:sds4}
\begin{tablenotes}
\item \scriptsize DS: \cite{li1999dimension}; IPCW-SIR: \cite{lu2011sufficient}; hMave: \cite{xia2010dimension}; Forward: forward regression; CP-SIR: the computational efficient approach; IR-CP: the counting process inverse regression; IR-Semi: the semiparametric inverse regression. For the proposed methods, ``$\widehat{\text{sd}}$'' is the bootstrap estimation of the standard deviation, and ``coverage'' is the coverage rate of the 95\% confidence internal.
\end{tablenotes}
\end{threeparttable}
\end{table}

\section{Additional results of TCGA data analysis}

This section contains additional results of the TCGA skin cutaneous melanoma data analysis. Table \ref{SKCM:dist} is the pairwise distance measure of the first direction estimated by all methods. This shows that the proposed methods mostly agree with each other. Figure \ref{fig:skcm.alt} presents the fitted models of competing approaches.

\begin{table}[htp]
    \footnotesize
    \caption{{Pairwise distance measures ($\times 10^2$) among all methods} }
    \center{
    \begin{tabular}{r r r r r r r r}
 &
 \begin{turn}{90} \begin{tabular}{l} Naive \end{tabular} \end{turn} & \begin{turn}{90} \begin{tabular}{l} DS \end{tabular} \end{turn} &
 \begin{turn}{90} \begin{tabular}{l} IPCW-SIR \end{tabular} \end{turn} & \begin{turn}{90} \begin{tabular}{l} hMave \end{tabular} \end{turn} &
 \begin{turn}{90} \begin{tabular}{l} \\ Forward \end{tabular} \end{turn} & \begin{turn}{90} \begin{tabular}{l} CP-SIR \end{tabular} \end{turn} &
 \begin{turn}{90} \begin{tabular}{l} IR-CP \end{tabular} \end{turn} \\
DS	&	141	&		&		&		&		&		&		\\
IPCW-SIR	&	63	&	141	&		&		&		&		&		\\
hMave	&	131	&	141	&	125	&		&		&		&		\\
Forward	&	141	&	87	&	141	&	135	&		&		&		\\
CP-SIR	&	140	&	92	&	140	&	129	&	31	&		&		\\
IR-CP	&	140	&	108	&	140	&	130	&	67	&	66	&		\\
IR-Semi	&	137	&	115	&	138	&	138	&	80	&	81	&	76	\\
\end{tabular}}\label{SKCM:dist}
\end{table}

\begin{figure}[htp]
\centering
\footnotesize
\caption{Estimated directions and fitted survival functions of alternative methods}{
    \begin{tabular}{cc}
    \includegraphics[height = 2.5in, width=0.45\textwidth, trim={5.5in 10.5in 4in 5in}, clip]{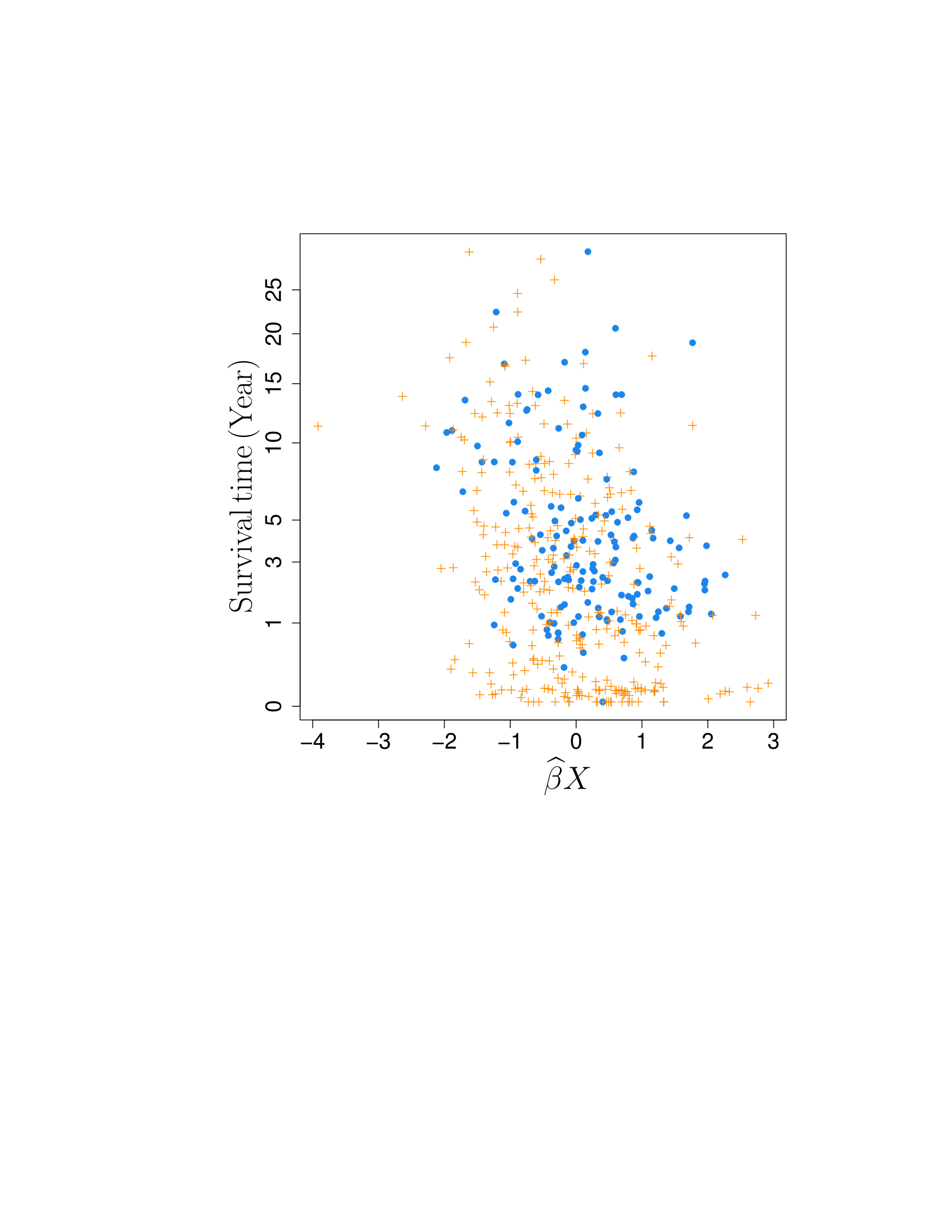} \hspace{0.05\textwidth}
    \includegraphics[height = 2.5in, width=0.45\textwidth, trim={5.5in 11in 3.5in 6in}, clip]{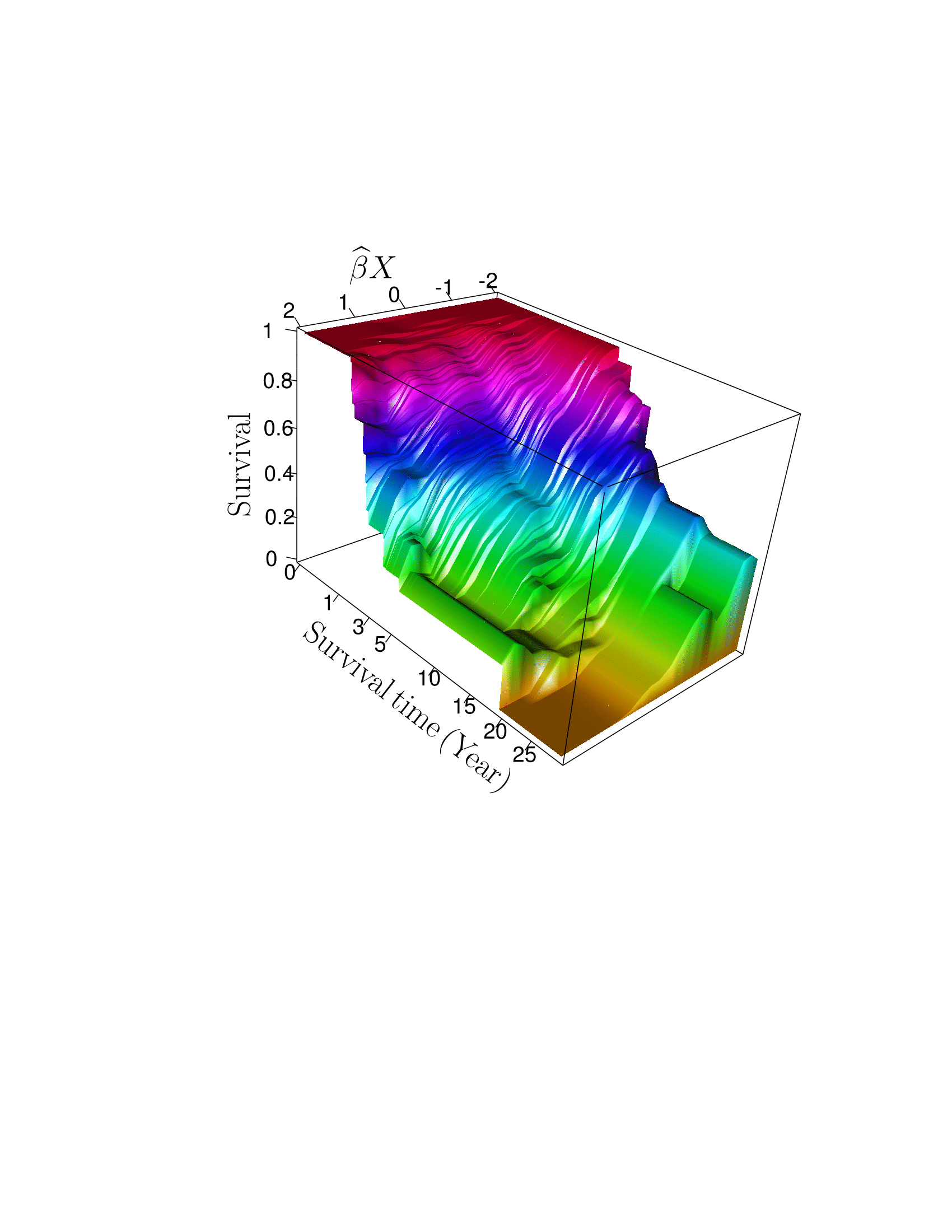}\\
    \includegraphics[height = 2.5in, width=0.45\textwidth, trim={5.5in 10.5in 4in 5in}, clip]{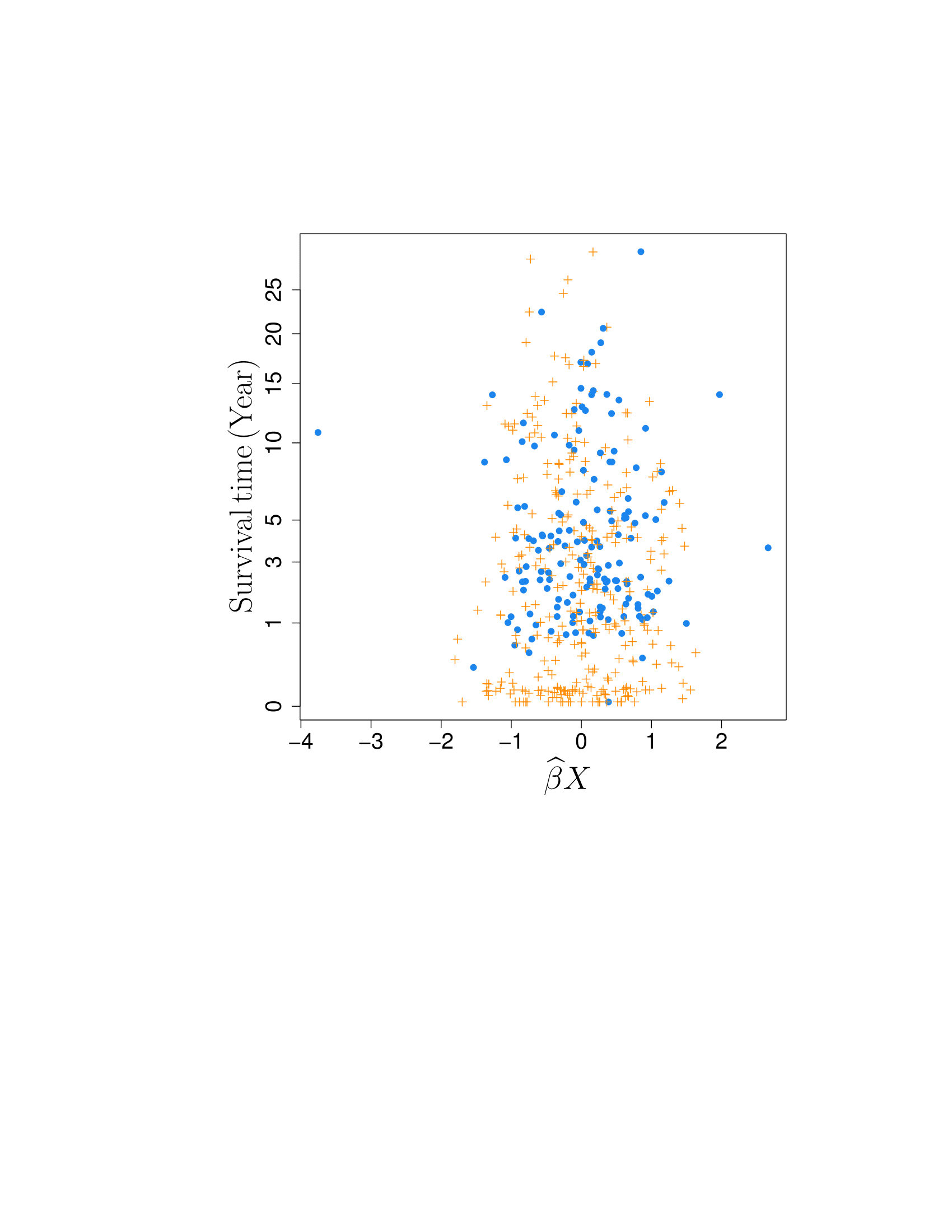} \hspace{0.05\textwidth}
    \includegraphics[height = 2.5in, width=0.45\textwidth, trim={5.5in 11in 3.5in 6in}, clip]{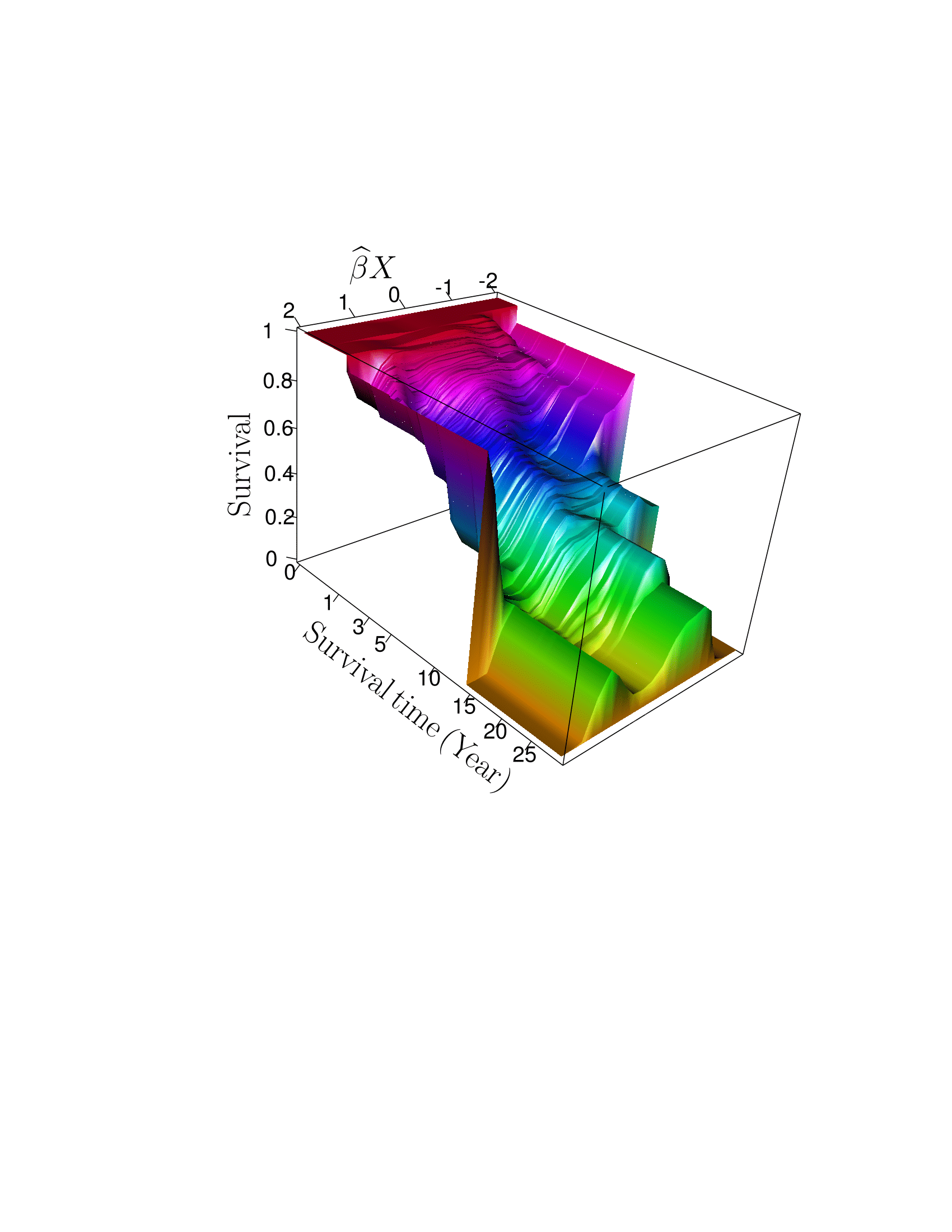}\\
    \includegraphics[height = 2.5in, width=0.45\textwidth, trim={5.5in 10.5in 4in 5in}, clip]{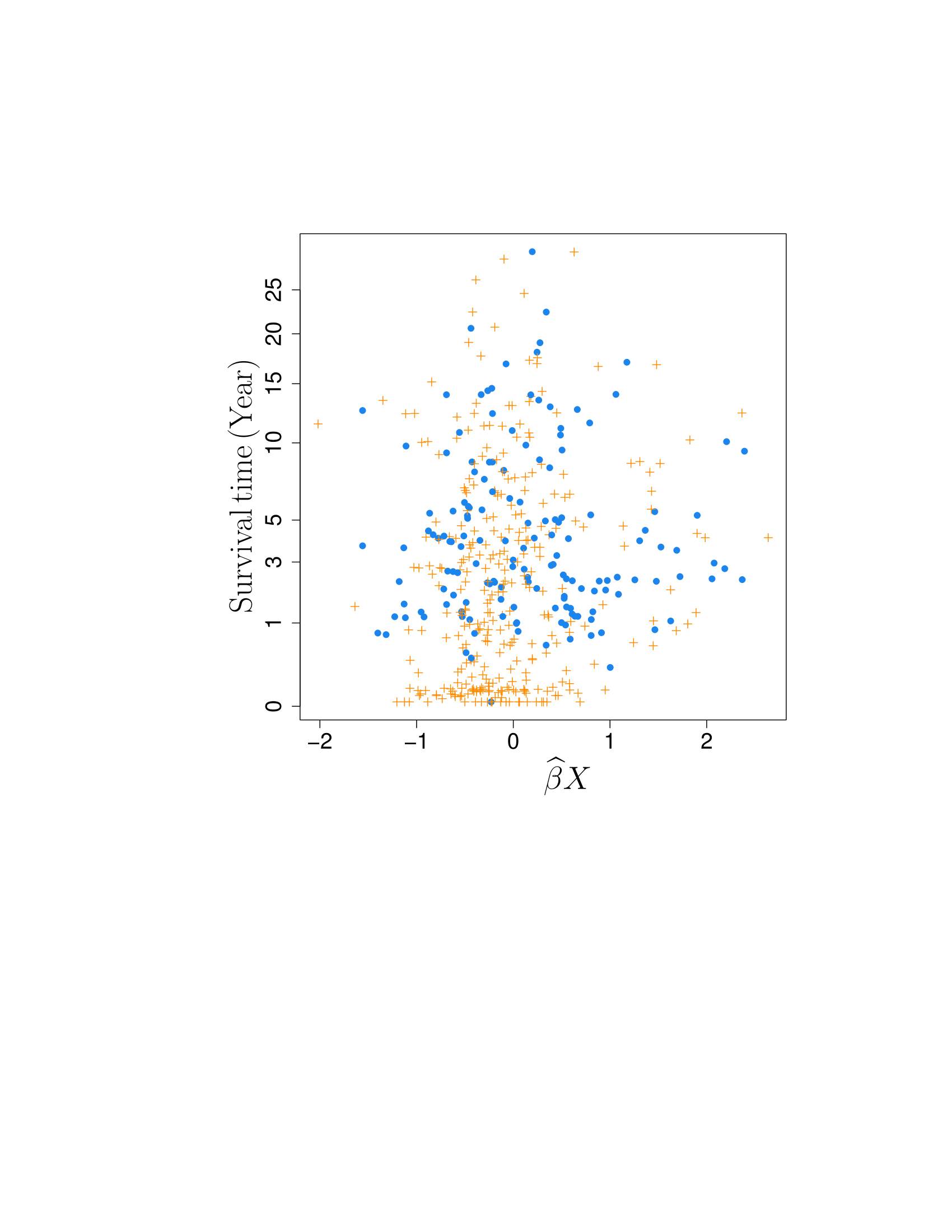} \hspace{0.05\textwidth}
    \includegraphics[height = 2.5in, width=0.45\textwidth, trim={5.5in 11in 3.5in 6in}, clip]{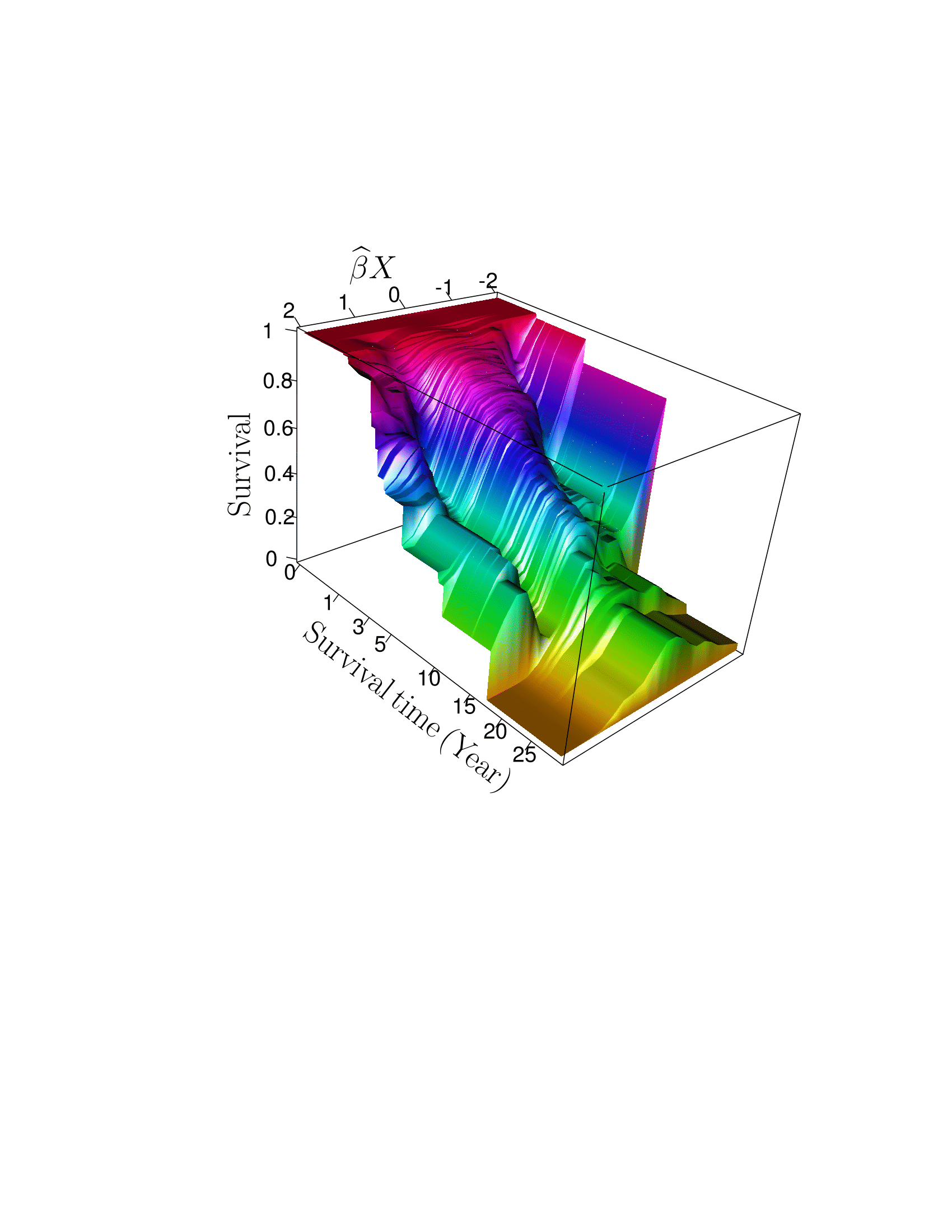}
    \end{tabular}}\label{fig:skcm.alt}
    \small \scriptsize First row: \cite{li1999dimension}; Second row: \cite{lu2011sufficient}; Third row: \cite{xia2010dimension}. For each row, the left panel is the projected direction versus the observed failure (blue dot) and censoring (orange $+$) times. The right panel is a nonparametric estimation of the survival function based on the projected direction.
\end{figure}

\bibliographystyle{biometrika}
\bibliography{CPDR_ref}

\end{document}